\documentclass{amsart}

\usepackage{amssymb,amsfonts,latexsym,amscd}
\usepackage{amsmath}
\usepackage[usenames]{color}
\usepackage[english]{babel}
\usepackage{epsfig}

\definecolor{Red}{rgb}{1,0.,0.}

\newtheorem{theorem}{Theorem}[section]

\newtheorem{lemma}{Lemma}[section]

\newtheorem{proposition}{Proposition}[section]
\newtheorem{remark}{Remark}[section]

\newtheorem{hypothesis}[theorem]{Hypothesis}

\newcommand{\gF}{{\mathfrak F}}  

\newcommand{\gD}{{\mathfrak D}}
\newcommand{\cD}{{\mathcal{D}}} 

\newcommand{{\adat}}{\mathrm{ad}_{A_t}}
\def\C{{\mathbb{C}}} 
\def\N{{\mathbb{N}}} 
\def\R{{\mathbb{R}}} 

\def\vp{\varphi}

\newcommand{\Aa}{{a_\epsilon}} 
\newcommand{\Ba}{b_{\ell, \epsilon}} 
\newcommand{\Ca}{c_{\ell, \epsilon}} 
\newcommand{\Ac}{{a^*_\epsilon}} 
\newcommand{\Bc}{b^*_{\ell, \epsilon}} 
\newcommand{\Cc}{c^*_{\ell, \epsilon}} 

\def\1{{\mathbf{1}}}

\renewcommand\d{\mathrm{d}}
\renewcommand\Re{\mathrm{Re}}  
\renewcommand\Im{\mathrm{Im}}  

\newcommand{\kt}{\tilde{K}} 
\newcommand{\ktt}{\tilde{\tilde{K}}} 
\newcommand{\tdd}{\tilde{D}_\delta}

\newcommand{\cb}{C_{\beta\eta}}
\newcommand{\bb}{B_{\beta\eta}}
\newcommand{\ct}{\tilde{C}_{\beta\eta}}
\newcommand{\bt}{\tilde{B}_{\beta\eta}}

\newcommand{\ant}{A_n^{(\tau)}}
\newcommand{\cnt}{\chi_n^{(\tau)}}

\newcommand{\lbr}{\langle b \rangle} 

\newcommand{\spec}{\mathrm{Spec}} 

\newcommand{\sig}{{\ \sigma}} 

\begin{document}

\title[Mathematical model of the weak interaction]{Spectral theory for a mathematical model of the
weak interaction: The decay of the intermediate vector bosons
$W^\pm$, II.}

\author[W.H.~Aschbacher]{Walter H. Aschbacher}
\author[J.-M.~Barbaroux]{Jean-Marie Barbaroux}
\author[J.~Faupin]{J\'er\'emy Faupin}
\author[J.-C.~Guillot]{Jean-Claude Guillot}

\address[W.H.~Aschbacher] {Centre de Math\'ematiques Appliqu\'ees, UMR 7641, \'Ecole
Polytechnique-CNRS, 91128 Palaiseau Cedex, France}
\email{walter.aschbacher@polytechnique.edu}

\address[J.-M.~Barbaroux]{Centre de Physique Th\'eorique, Luminy Case 907, 13288
 Marseille Cedex~9, France and D\'epartement de Math\'ematiques,
 Universit\'e du Sud Toulon-Var, 83957 La
 Garde Cedex, France}
\email{barbarou@univ-tln.fr}

\address[J.~Faupin] {Institut de Math\'ematiques de Bordeaux, UMR-CNRS 5251,
Universit\'e de Bordeaux I, 351 cours de la Lib\'eration, 33405
Talence Cedex, France}
\email{jeremy.faupin@math.u-bordeaux1.fr}

\address[J.-C.~Guillot]{Centre de Math\'ematiques Appliqu\'ees, UMR 7641, \'Ecole
Polytechnique-CNRS, 91128 Palaiseau Cedex, France}
\email{jean-claude.guillot@polytechnique.edu}

\maketitle

\begin{center}
\textit{In memory of Pierre Duclos.}
\end{center}


\begin{abstract}
We do the spectral analysis of the Hamiltonian for the weak
leptonic decay of the gauge bosons $W^\pm$. Using Mourre theory,
it is shown that the spectrum between the unique ground state and
the first threshold is purely absolutely continuous. Neither sharp
neutrino high energy cutoff nor infrared regularization are
assumed.
\end{abstract}


\section{Introduction}
We study a mathematical model for the weak decay of the
intermediate vector bosons $W^\pm$ into the full family of
leptons. The full family of leptons involves the electron $e^-$
and the positron $e^+$, together with the associated neutrino
$\nu_e$ and antineutrino $\bar\nu_e$, the muons $\mu^-$ and
$\mu^+$ together with the associated neutrino $\nu_\mu$ and
antineutrino $\bar\nu_\mu$ and the tau leptons $\tau^-$ and
$\tau^+$ together with the associated neutrino $\nu_\tau$ and
antineutrino $\bar\nu_\tau$.
The model is patterned according to the Standard Model in Quantum
Field Theory (see \cite{GreinerMuller1989, WeinbergII2005}).

A representative and well-known example of this general process is
the decay of the gauge boson $W^-$ into an electron and an
antineutrino of the electron that occurs in $\beta$-decay,
\begin{equation}\label{eq:ivb}
   W^-\rightarrow e^- +{\bar \nu}_e.
\end{equation}
If we include the corresponding
antiparticles in the process \eqref{eq:ivb},  the interaction described in the Schr\"odinger
representation is formally given by (see
\cite[(4.139)]{GreinerMuller1989} and
\cite[(21.3.20)]{WeinbergII2005})
\begin{equation}\nonumber
  I = \int \d^3 \!x\, \overline{\Psi_e}(x) \gamma^\alpha
  (1-\gamma_5)\Psi_{\nu_e}(x) W_\alpha(x) + \int \d^3\!x\,
  \overline{\Psi_{\nu_e}}(x) \gamma^\alpha (1-\gamma_5)
  \Psi_e(x) W_\alpha(x)^* ,
\end{equation}
where $\gamma^\alpha$, $\alpha=0,1,2,3$, and $\gamma_5$ are the
Dirac matrices, $\Psi_.(x)$ and $\overline{\Psi_.}(x)$ are the
Dirac fields for $e_\pm$, $\nu_e$, and $\bar\nu_e$, and $W_\alpha$
are the boson fields (see \cite[\S 5.3]{WeinbergI2005}) given by
\begin{equation}\nonumber
\begin{split}
 \Psi_e(x) = & {(2\pi)}^{-\frac32}\!\!\! \sum_{s=\pm\frac12}\!
 \int \d^3\!p
 \Big(
 b_{e,+}(p,s) \frac{u(p,s)}{(2(|p|^2\! +\! m_e^2)^\frac12)^\frac12} \mathrm{e}^{ip.x}\\
 &
 \quad\quad\quad\quad\quad+ b_{e,-}^* (p,s)
 \frac{v(p,s)}{(2(|p|^2\! +\! m_e^2)^\frac12)^\frac12}
 \mathrm{e}^{- i p .x}
 \Big), \\
 \Psi_{\nu_e}(x) = &{(2\pi)}^{-\frac32}\!\!\! \sum_{s=\pm\frac12}\!
 \int \d^3\!p\, \Big(c_{e,+}(p,s) \frac{u(p,s)}{(2|p|)^\frac12}
 \mathrm{e}^{ip.x}
 + c_{e,-}^* (p,s) \frac{v(p,s)}{(2|p|)^\frac12} \mathrm{e}^{- i p .
 x}\Big)\, , \\
 \overline{\Psi_e}(x) = &\Psi_e(x)^\dagger \gamma^0\, ,\quad
 \overline{\Psi_{\nu_e}}(x) = \Psi_{\nu_e}(x)^\dagger \gamma^0\, ,
\end{split}
\end{equation}
and
\begin{equation}\nonumber
\begin{split}
 W_\alpha(x) = {(2\pi)}^{-\frac32} \sum_{\lambda=-1,0,1}
 \int \frac{\d^3\!k}{(2 (|k|^2 \!+\! m_W^2)^\frac12)^\frac12}
 & \Big(\epsilon_\alpha(k,\lambda)
 a_+(k,\lambda)
 \mathrm{e}^{i k.x} \\
 & + \epsilon_\alpha^*(k,\lambda) a_-^*(k,\lambda)
 \mathrm{e}^{- i k.x}\Big)\, .
 \end{split}
\end{equation}

Here $m_e>0$ is the mass of the electron and $u(p,s)/(2(|p|^2\!
+\! m_e^2)^{1/2})^{1/2}$ and $v(p,s)/(2(|p|^2\! +\!
m_e^2)^{1/2})^{1/2}$ are the normalized solutions to the Dirac
equation (see \cite[Appendix]{GreinerMuller1989}), $m_W>0$ is the
mass of the bosons $W^\pm$ and the vectors
$\epsilon_\alpha(k,\lambda)$ are the polarizations of the massive
spin~1 bosons (see \cite[Section~5.2]{WeinbergI2005}), and as
follows from the Standard Model, neutrinos and antineutrinos are
considered to be massless particles.

The operators $b_{e,+}(p,s)$ and $b_{e,+}^*(p,s)$ (respectively
$c_{\nu_e,+}(p,s)$ and $c_{\nu_e,+}^*(p,s)$), are the annihilation
and creation operators for the electrons (respectively for the
neutrinos associated with the electrons), satisfying the
anticommutation relations. The index $-$ in $b_{e,-}(p,s)$,
$b_{e,-}^*(p,s)$, $c_{\nu_e,-}(p,s)$ and $c_{\nu_e,-}^*(p,s)$ are
used to denote the annihilation and creation operators of the
corresponding antiparticles.
The operators $a_+(k,\lambda)$ and $a_+^*(k,\lambda)$
(respectively $a_-(k,\lambda)$ and $a_-^*(k,\lambda)$) are the
annihilation and creation operators for the bosons $W^-$
(respectively $W^+$) satisfying the canonical commutation
relations.

If one considers the full interaction describing the decay of the
gauge bosons $W^\pm$ into leptons
(see \cite[(4.139)]{GreinerMuller1989}) and if one formally expands
this interaction with respect to products of creation and
annihilation operators, we are left with a finite sum of terms
associated with kernels of the form
\begin{equation}\nonumber
  \delta(p_1 + p_2 - k) g(p_1,\ p_2,\, k)\ .
\end{equation}
The $\delta$-distributions that occur here shall be approximated
by square integrable functions. Therefore, in this article, the
interaction for the weak decay of $W^\pm$ into the full family of
leptons will be described in terms of annihilation and creation
operators together with kernels which are square integrable with
respect to momenta (see \eqref{eq:hypothese-1} and
\eqref{eq:HI-quadraticform-1}-\eqref{eq:HI-quadraticform-3}).

Under this assumption, the total Hamiltonian, which is the sum of
the free energy of the particles (see \eqref{eq:def-free}) and of
the interaction, is a well-defined self-adjoint operator in the
Fock space for the leptons and the vector bosons
(Theorem~\ref{thm:self-adjointness}). This allows us to study its
spectral properties.

Among the four fundamental interactions known up to now, the weak
interaction does not generate bound states, which is not the case
for the strong, the electromagnetic and the gravitational interactions.
Thus we are expecting that the spectrum of the Hamiltonian
associated with every model of weak decays is purely absolutely
continuous above the ground state energy.

With additional assumptions on the kernels that are fulfilled by
the model described in theoretical physics, we can prove
(Theorem~\ref{thm:main-1}; see also
\cite[Theorem~3.3]{BarbarouxGuillot2009}) that the total
Hamiltonian has a unique ground state in  Fock space for a
sufficiently small coupling constant, corresponding to the dressed
vacuum. The strategy for proving existence of a unique ground
state dates back to the early works of Bach, Fr\"ohlich, and Sigal
\cite{Bachetal1999S} and Griesemer, Lieb, and Loss \cite{GLL}, for
the Pauli-Fierz model of non-relativistic QED. Our proofs follow
these techniques as adapted to a model of quantum
electrodynamics \cite{Barbarouxetal2004a, Barbarouxetal2004,
DimassiGuillot} and a model of the Fermi weak interactions
\cite{Amouretal2007}.

Moreover, under natural regularity assumptions on the kernels, we
establish a Mourre estimate (Theorem~\ref{thm:Mourre}) and a
limiting absorption principle (Theorem~\ref{thm:LAP}) for any
spectral interval above the energy of the ground state and below
the mass of the electron, for small enough coupling constants. As
a consequence, the spectrum between the unique ground state and
the first threshold is shown to be purely absolutely continuous
(Theorem~\ref{thm:main-2}).

To achieve the spectral analysis above the ground state energy,
our methods are taken largely from \cite{Bachetal2006},
\cite{Frohlichetal2008}, and \cite{Chenetal2009}. More precisely, we begin with approximating the total Hamiltonian $H$ by a cutoff Hamiltonian $H_\sigma$ which has the property that the interaction between the massive particles and the neutrinos or antineutrinos of energies $\le \sigma$ has been suppressed. The restriction of $H_\sigma$ to the Fock space for the massive particles together with the neutrinos and antineutrinos of energies $\ge \sigma$ is 
 denoted by $H^\sigma$ in this paper. Adapting the method of \cite{Bachetal2006}, we prove that, for some suitable sequence $\sigma_n \to 0$, the Hamiltonian $H^{\sigma_n}$ has a gap of size $c\sigma_n$ in its spectrum above its ground state energy for all $n \in \mathbb{N}$. In contrast to \cite{BarbarouxGuillot2009}, we do not require a sharp neutrino high energy cutoff here.

Next, as in \cite{Frohlichetal2008}, \cite{BarbarouxGuillot2009}
and \cite{Chenetal2009}, we use the gap  property in combination
with the conjugate operator method developed in
\cite{Amreinetal1996} and \cite{Sahbani1997} in order to establish
a limiting absorption principle near the ground state energy of
$H$. In \cite{BarbarouxGuillot2009}, the chosen conjugate operator
is the generator of dilatations in the Fock space for neutrinos
and antineutrinos. As a consequence, an infrared regularization is
assumed in \cite{BarbarouxGuillot2009} in order to be able to
implement the strategy of \cite{Frohlichetal2008}. Let us mention
that no infrared regularization is required in
\cite{Frohlichetal2008} since, for the model of non-relativistic
QED with a fixed nucleus which is studied in that paper, a unitary
Pauli-Fierz transformation can be applied with the effect of
regularizing the infrared behavior of the interaction.

In the present paper, we choose a conjugate operator which is the
generator of dilatations in the Fock space for neutrinos and
antineutrinos \emph{with a cutoff in the momentum variable}. Hence
our conjugate operator only affects the massless particles of low
energies.  A similar choice is made in \cite{Chenetal2009}, where
the Pauli-Fierz model of non-relativistic QED for a free electron
at a fixed total momentum is studied. Due to the complicated
structure of the interaction operator in this context, the authors
in \cite{Chenetal2009} make use of some Feshbach-Schur map before
proving a Mourre estimate for an effective Hamiltonian. Here we do
not need to apply such a map, and we prove a Mourre estimate
directly for $H$. Compared with \cite{Frohlichetal2008}, our
method involves further estimates, which allows us to avoid any
infrared regularization.

As mentioned before, some of the basic results of this article have
been previously stated and proved, under stronger assumptions, in
\cite{BarbarouxGuillot2009a, BarbarouxGuillot2009}.
The main achievement of this paper in comparison with \cite{BarbarouxGuillot2009} is that no sharp neutrino high
energy cutoff and no infrared regularization are assumed here.

The nature of the spectrum above the first threshold and the
scattering theory of this model remain to be studied elsewhere.\\

The paper is organized as follows. In the next section, we give a
precise definition of the Hamiltonian.
Section~\ref{S:spectral-properties} is devoted to the statements
of the main spectral properties. In
sections~\ref{S-gap}-\ref{S-LAP}, we establish the results
necessary to apply Mourre theory, namely, we derive a gap
condition, a Mourre estimate, local $C^2$-regularity of the
Hamiltonian, and a limiting absorption principle.
Section~\ref{S-proof} details the proof of
Theorem~\ref{thm:main-2} on absolute continuity of the spectrum.
Eventually, in Appendix~\ref{appendix}, we state and prove several
technical lemmata.
For the sake of clarity, all proofs in sections~\ref{S-gap} to
\ref{S-proof} and in Appendix~\ref{appendix} are given for the
particular process depicted in \eqref{eq:ivb}. The general
situation can be recovered by a straightforward generalization.

\setcounter{equation}{0}

\section{Definition of the model}\label{S2}

\setcounter{equation}{0}

According to the Standard Model, the weak decay of the intermediate bosons $W^+$ and $W^-$ involves the full family of leptons together with the bosons themselves (see \cite[Formula
(4.139)]{GreinerMuller1989} and \cite{WeinbergII2005}).
As mentioned in the Introduction, the full family of leptons consists of the electron $e^-$, the
muon $\mu^-$, the tau lepton $\tau^-$, their associated neutrinos
$\nu_e$, $\nu_\mu$, $\nu_\tau$ and all their antiparticles $e^+$,
$\mu^+$, $\tau^+$, $\bar\nu_e$, $\bar\nu_\mu$, and $\bar\nu_\tau$.
In the Standard Model, neutrinos and antineutrinos are massless
particles with helicity $-1/2$ and $+1/2$, respectively. Here we
shall assume that both neutrinos and antineutrinos have helicity
$\pm 1/2$.

The mathematical model for the weak decay of the vector bosons
$W^\pm$ is defined as follows.
The index $\ell\in\{1,2,3\}$ denotes each species of leptons:
$\ell=1$ denotes the electron $e^-$, the positron $e^+$ and the
associated neutrinos $\nu_e$, $\bar\nu_e$; $\ell=2$ denotes the
muons $\mu^-$, $\mu^+$ and the associated neutrinos $\nu_\mu$ and
$\bar\nu_\mu$; and $\ell=3$ denotes the tau-leptons and the
associated neutrinos $\nu_\tau$ and $\bar\nu_\tau$.

Let $\xi_1=(p_1,\ s_1)$ be the quantum variables of a massive
lepton, where $p_1\in\R^3$ and $s_1\in\{-1/2,\ 1/2\}$ is the spin
polarization of particles and antiparticles. Let $\xi_2=(p_2,\
s_2)$ be the quantum variables of a massless lepton, where
$p_2\in\R^3$ and $s_2\in\{-1/2,\ 1/2\}$ is the helicity of
particles and antiparticles, and, finally, let $\xi_3=(k,\
\lambda)$ be the quantum variables of the spin $1$ bosons $W^+$
and $W^-$, where $k\in\R^3$ and $\lambda\in\{-1,\ 0,\ 1\}$ accounts
for the polarization of the vector bosons (see
\cite[section~5.2]{WeinbergI2005}).

Moreover, we set $\Sigma_1 = \R^3\times\{-1/2,\ 1/2\}$ for the configuration
space of the leptons and $\Sigma_2 = \R^3\times\{-1,\ 0,\ 1\}$ for
the bosons. Thus $L^2(\Sigma_1)$ is the Hilbert space of each
lepton and $L^2(\Sigma_2)$ is the Hilbert space of each boson. In
the sequel, we shall use the notations $\int_{\Sigma_1} \d \xi =
\sum_{s=+\frac12, -\frac12} \int \d^3  p$ and $\int_{\Sigma_2} \d
\xi = \sum_{\lambda=0,1,-1} \int \d^3 k$.

The Hilbert space for the weak decay of the vector bosons $W^\pm$
is the Fock space for leptons and bosons describing the set of
states with indefinite number of particles or antiparticles which
we define below.

For every $\ell$, $\gF_\ell$ is the fermionic Fock space for the
corresponding species of leptons including the massive particle
and antiparticle together with the associated neutrino and
antineutrino, i.e., for $\ell=1,2,3$, 
\begin{equation}\nonumber
 \gF_\ell
 = \bigotimes^4 \gF_a(L^2(\Sigma_1))
 = \bigotimes^4\left(\oplus_{n=0}^\infty \otimes_a^n
 L^2(\Sigma_1)\right)\, ,
\end{equation}
where $\otimes_a^n$ denotes the antisymmetric $n$-th tensor
product and $\otimes_a^0 L^2(\Sigma_1) = \C$.
The fermionic Fock space $\gF_L$ for the leptons is
\begin{equation}\nonumber
 \gF_L = \otimes_{\ell=1}^3 \, \gF_\ell\ ,
\end{equation}
and the bosonic Fock space $\gF_W$ for the vector bosons $W^+$ and
$W^-$ reads
\begin{equation}\nonumber
 \gF_W
 = \bigotimes^2\gF_s(L^2(\Sigma_2))
 = \bigotimes^2 \left(\oplus_{n=0}^\infty \otimes_s^n
 L^2(\Sigma_2)\right)
 \ ,
\end{equation}
where $\otimes_s^n$ denotes the symmetric $n$-th tensor product
and $\otimes_s^0 L^2(\Sigma_2) = \C$.
Finally, the Fock space for the weak decay of the vector bosons $W^+$ and
$W^-$ is thus
\begin{equation}\nonumber
 \gF = \gF_L \otimes\gF_W\ .
\end{equation}

Furthermore, for each $\ell=1,2,3$, $\Ba(\xi_1)$ (resp. $\Bc(\xi_1)$) is the
annihilation (resp. creation) operator for the corresponding
species of massive particle if $\epsilon=+$ and for the
corresponding species of massive antiparticle if $\epsilon=-$.
Similarly, for each $\ell=1,2,3$, $\Ca(\xi_2)$ (resp.
$\Cc(\xi_2)$) is the annihilation (resp. creation) operator for
the corresponding species of neutrino if $\epsilon=+$ and for
the corresponding species of antineutrino if $\epsilon=-$.
Finally, the operator $\Aa(\xi_3)$ (resp. $\Ac(\xi_3)$) is the
annihilation (resp. creation) operator for the boson $W^-$ if
$\epsilon=+$, and for the boson $W^+$ if $\epsilon=-$. The
operators $b_{\ell,\epsilon}(\xi_1)$,
$b_{\ell,\epsilon}^*(\xi_1)$, $c_{\ell,\epsilon}(\xi_2)$, and
$c_{\ell,\epsilon}^*(\xi_2)$ fulfil the usual canonical
anticommutation relations (CAR), whereas $\Aa(\xi_3)$ and
$\Ac(\xi_3)$ fulfil the canonical commutation relation (CCR), see
e.g. \cite{WeinbergI2005}. Moreover, the $a$'s commute with the
$b$'s and the $c$'s.
In addition, following the convention described in
\cite[section~4.1]{WeinbergI2005} and
\cite[section~4.2]{WeinbergI2005}, we will assume that fermionic
creation and annihilation operators of different species of
leptons anticommute (see e.g.
\cite{BarbarouxGuillot2009} arXiv for explicit definitions).
Therefore, the following canonical anticommutation and commutation
relations hold,
\begin{equation}\nonumber
\begin{split}
 &\{ b_{\ell, \epsilon}(\xi_1), b^*_{\ell', \epsilon'}(\xi_1')\} =
 \delta_{\ell \ell'}\delta_{\epsilon \epsilon'} \delta(\xi_1 - \xi_1') \ ,\\
 &\{ c_{\ell, \epsilon}(\xi_2), c^*_{\ell', \epsilon'}(\xi_2')\} =
 \delta_{\ell \ell'}\delta_{\epsilon \epsilon'} \delta(\xi_2 - \xi_2')\ , \\
 &[ a_{\epsilon}(\xi_3), a^*_{\epsilon'}(\xi_3')] =
 \delta_{\epsilon \epsilon'} \delta(\xi_3 - \xi_3') \ ,\\
 &\{ b_{\ell, \epsilon}(\xi_1), b_{\ell', \epsilon'}(\xi_1')\}
 = \{ c_{\ell, \epsilon}(\xi_2), c_{\ell', \epsilon'}
 (\xi_2')\} =0\ ,\\
 & [ a_{\epsilon}(\xi_3), a_{\epsilon'}(\xi_3') ] = 0\ ,\\
 &\{ b_{\ell, \epsilon}(\xi_1), c_{\ell', \epsilon'}(\xi_2)\}
 = \{ b_{\ell, \epsilon}(\xi_1), c^*_{\ell',
 \epsilon'}(\xi_2)\}=0 \ , \\
 &[ b_{\ell, \epsilon}(\xi_1), a_{\epsilon'}(\xi_3)]
 = [ b_{\ell, \epsilon}(\xi_1), a^*_{\epsilon'}(\xi_3)]
 = [ c_{\ell, \epsilon}(\xi_2), a_{\epsilon'}(\xi_3)]
 = [ c_{\ell, \epsilon}(\xi_2), a^*_{\epsilon'}(\xi_3)] = 0
 \ ,
\end{split}
\end{equation}
where $\{b, b'\} = bb' + b'b$ and $[a,a'] = aa' - a'a$. Moreover,
we recall that,  for $\vp\in L^2(\Sigma_1)$, the operators
\begin{equation}\nonumber
\begin{split}
  & b_{\ell, \epsilon}(\vp) = \int_{\Sigma_1} b_{\ell,
  \epsilon}(\xi) \overline{\vp(\xi)} \d \xi,\quad
  c_{\ell, \epsilon}(\vp) = \int_{\Sigma_1} c_{\ell,
  \epsilon}(\xi)\overline{\vp(\xi)} \d \xi\ , \\
  & b^*_{\ell, \epsilon}(\vp) = \int_{\Sigma_1} b^*_{\ell,
  \epsilon}(\xi) {\vp(\xi)} \d \xi,\quad
  c^*_{\ell, \epsilon}(\vp) = \int_{\Sigma_1} c^*_{\ell,
  \epsilon}(\xi){\vp(\xi)} \d \xi\ ,
\end{split}
\end{equation}
are bounded operators on $\gF$ satisfying
\begin{equation}\nonumber\\
  \| b^\sharp_{\ell,\epsilon}(\vp)\| = \|
  c^\sharp_{\ell,\epsilon}(\vp)\| = \|\vp\|_{L^2}\ ,
\end{equation}
where $b^\sharp$ (resp. $c^\sharp$) is b (resp. $c$) or $b^*$
(resp. $c^*$).

The free Hamiltonian $H_0$ is given by
\begin{equation}\label{eq:def-free}
\begin{split}
 H_0
 & = \sum_{\ell=1}^3 \sum_{\epsilon=\pm} \int
 w_\ell^{(1)}(\xi_1) b^*_{\ell,\epsilon}(\xi_1)
 b_{\ell,\epsilon}(\xi_1) \d \xi_1
 + \sum_{\ell=1}^3 \sum_{\epsilon=\pm}\int
 w_\ell^{(2)}(\xi_2) c^*_{\ell,\epsilon}(\xi_2)
 c_{\ell,\epsilon}(\xi_2) \d \xi_2 \\
 & + \sum_{\epsilon=\pm} \int w^{(3)}(\xi_3) a^*_\epsilon(\xi_3)
 a_\epsilon(\xi_3) \d \xi_3\ ,
\end{split}
\end{equation}
where the free relativistic energy of the massive leptons, the
neutrinos, and the bosons are respectively given by
\begin{equation}\nonumber
\begin{split}
 w_\ell^{(1)}(\xi_1) = (|p_1|^2 + m_\ell^2)^\frac12,\
 w_\ell^{(2)}(\xi_2) = |p_2|\mbox{, and }
 w^{(3)}(\xi_3) = (|k|^2 + m^2_W)^\frac12\ .
\end{split}
\end{equation}
Here $m_\ell$ is the mass of the lepton $\ell$ and $m_W$ is the
mass of the bosons, with $m_1 < m_2 < m_3 < m_W$.

The interaction $H_I$ is described in terms of annihilation and
creation operators together with kernels $G^{(\alpha)}_{\ell,
\epsilon, \epsilon'} (.,.,.)$ ($\alpha=1,2$).

As emphasized previously, each kernel $G_{\ell, \epsilon,
\epsilon'}^{(\alpha)}(\xi_1, \xi_2, \xi_3)$, computed in
theoretical physics, contains a $\delta$-distribution because of
the conservation of the momentum (see \cite{GreinerMuller1989},
\cite[section~4.4]{WeinbergI2005}). Here, we approximate the
singular kernels by square integrable functions. Therefore, we
assume the following

\begin{hypothesis}\label{hypothesis:2.1}
For $\alpha=1,2$, $\ell=1,2,3$, $\epsilon, \epsilon' = \pm$, we
assume
\begin{equation}\label{eq:hypothese-1}
 G^{(\alpha)}_{\ell, \epsilon, \epsilon'} (\xi_1, \xi_2, \xi_3)\in
 L^2(\Sigma_1\times \Sigma_1\times\Sigma_2)\ .
\end{equation}
\end{hypothesis}

Based on \cite[p159, (4.139)]{GreinerMuller1989} and \cite[p308,
(21.3.20)]{WeinbergII2005}, we define the interaction  as
\begin{equation}\label{eq:HI-quadraticform-1}
  H_I = H_I^{(1)} + H_I^{(2)}\ ,
\end{equation}
where
\begin{equation}\label{eq:HI-quadraticform-2}
\begin{split}
  H_I^{(1)} = & \sum_{\ell=1}^3 \sum_{\epsilon\neq\epsilon'}
  \int G^{(1)}_{\ell,\epsilon,\epsilon'} (\xi_1, \xi_2,\xi_3)
  b^*_{\ell,\epsilon}(\xi_1) c^*_{\ell, \epsilon'}(\xi_2)
  a_\epsilon(\xi_3) \d \xi_1 \d \xi_2 \d \xi_3 \\
  & + \sum_{\ell=1}^3 \sum_{\epsilon\neq \epsilon'}
  \int\overline{G^{(1)}_{\ell, \epsilon,\epsilon'} (\xi_1,
  \xi_2,\xi_3)} a^*_\epsilon(\xi_3)
  c_{\ell,\epsilon'} (\xi_2) b_{\ell,\epsilon}(\xi_1) \d \xi_1 \d
  \xi_2 \d \xi_3 \ ,
\end{split}
\end{equation}
\begin{equation}\label{eq:HI-quadraticform-3}
\begin{split}
  H_I^{(2)} = & \sum_{\ell=1}^3 \sum_{\epsilon\neq\epsilon'}
  \int G^{(2)}_{\ell,\epsilon,\epsilon'} (\xi_1, \xi_2,\xi_3)
  b^*_{\ell,\epsilon}(\xi_1) c^*_{\ell, \epsilon'}(\xi_2)
  a^*_\epsilon(\xi_3) \d \xi_1 \d \xi_2 \d \xi_3 \\
  & + \sum_{\ell=1}^3 \sum_{\epsilon\neq \epsilon'}
  \int\overline{G^{(2)}_{\ell, \epsilon,\epsilon'} (\xi_1,
  \xi_2,\xi_3)}
  a_\epsilon(\xi_3)
  c_{\ell,\epsilon'}(\xi_2)
  b_{\ell,\epsilon}(\xi_1)
    \d \xi_1 \d
  \xi_2 \d \xi_3 \ .
\end{split}
\end{equation}
The operator $H_I^{(1)}$ describes the decay of the bosons $W^+$
and $W^-$ into leptons, and $H_I^{(2)}$ is responsible for the
fact that the bare vacuum will not be an eigenvector of the total
Hamiltonian, as expected from physics.

For $\ell=1,2,3$, all terms in $H_I^{(1)}$ and $H_I^{(2)}$ are
well defined as quadratic forms on the set of finite particle
states consisting of smooth wave functions. According to
\cite[Theorem X.24]{ReedSimon1975} (see details in
\cite{BarbarouxGuillot2009}), one can construct a closed operator
associated with the quadratic form defined by
\eqref{eq:HI-quadraticform-1}-\eqref{eq:HI-quadraticform-3}.

The total Hamiltonian is thus ($g$ is a coupling constant),
\begin{equation}\nonumber
  H = H_0 + g H_I,\quad g>0\ .
\end{equation}

\begin{theorem}\label{thm:self-adjointness}
Let $g_1>0$ be such that
\begin{equation}\nonumber
 \frac{6 g_1^2}{m_W} \left(\frac{1}{m_1^2} +1\right) \sum_{\alpha=1,2}
 \sum_{\ell=1}^3 \sum_{\epsilon\neq\epsilon'} \|
 G^{(\alpha)}_{\ell,\epsilon,\epsilon'}
 \|^2_{L^2(\Sigma_1\times\Sigma_1\times\Sigma_2)} <1\ .
\end{equation}
Then, for every $g$ satisfying $g\leq g_1$, $H$ is a self-adjoint
operator in $\gF$ with domain $\cD(H)=\cD(H_0)$.
\end{theorem}
%
%
%
This result has been proven in
~\cite[Theorem~2.6]{BarbarouxGuillot2009} (with a prefactor $2$
missing).

\section{Location of the spectrum, existence of a ground state,
absolutely continuous spectrum, and dynamical
properties}\label{S:spectral-properties}

\setcounter{equation}{0}

In the sequel, we shall make some of the following additional
assumptions on the kernels
$G^{(\alpha)}_{\ell,\epsilon,\epsilon'}$.

\begin{hypothesis}\label{hypothesis:3.1} There exist $\kt(G)$
and $\ktt(G)$ such that for $\alpha=1,2,\ \ell=1,2,3,\
\epsilon,\epsilon'=\pm $, $i,j=1,2,3$, and $\sigma\geq 0$,
\begin{equation*}
\begin{split}
 \mbox{(i)}\quad &
 \int_{\Sigma_1\times\Sigma_1\times\Sigma_2}
 \frac{|
 G^{(\alpha)}_{\ell,\epsilon,\epsilon'}(\xi_1,\xi_2,\xi_3)|^2}{|p_2|^2}
 \d\xi_1 \d\xi_2 \d\xi_3 <\infty\ ,
\end{split}
\end{equation*}
\begin{equation*}
\begin{split}
 \mbox{(ii)}\quad &
 \left(\int_{\Sigma_1\times\{|p_2|\leq \sigma\}\times\Sigma_2}
 |G^{(\alpha)}_{\ell,\epsilon,\epsilon'}(\xi_1,\xi_2,\xi_3)|^2
 \d\xi_1 \d\xi_2 \d\xi_3\right)^\frac12 \leq \kt(G)\, \sigma \ ,
\end{split}
\end{equation*}
\begin{equation*}
\begin{split}
 \mbox{(iii-a)}\quad &
 (p_2\cdot\nabla_{p_2})
 G^{(\alpha)}_{\ell,\epsilon,\epsilon'}(.,.,.)\in
 L^2(\Sigma_1\times\Sigma_1\times\Sigma_2)\, \mbox{and}\ \\
 & \int_{\Sigma_1\times\{|p_2|\leq \sigma\}\times\Sigma_2}
 \left| [ (p_2\cdot\nabla_{p_2})
 G^{(\alpha)}_{\ell,\epsilon,\epsilon'}](\xi_1,\xi_2,\xi_3)\right|^2
 \d\xi_1 \d\xi_2 \d\xi_3 < \ktt(G)\, \sigma ,\\
 \mbox{(iii-b)}\quad & \int_{\Sigma_1\times\Sigma_1\times\Sigma_2}
 p_{2,i}^2\, p_{2,j}^2
 \left|
 \frac{\partial^2 G^{(\alpha)}_{\ell,\epsilon,\epsilon'}}
 {\partial p_{2,i} \partial p_{2,j}}(\xi_1,\xi_2,\xi_3)\right|^2
 \d\xi_1 \d\xi_2 \d\xi_3 < \infty\ .
\end{split}
\end{equation*}
\end{hypothesis}

\begin{remark}
Note that Hypothesis 3.1. {\it (i)} is stronger than 
Hypothesis 3.1. {\it (ii)}, of course.
\end{remark}

Our first main result is devoted to the existence of a ground
state for $H$ together with the location of the spectrum of $H$
and of its absolutely continuous spectrum.

\begin{theorem}\label{thm:main-1}
Assume that the kernels $G^{(\alpha)}_{\ell,\epsilon,\epsilon'}$
satisfy Hypotheses~\ref{hypothesis:2.1} and
\ref{hypothesis:3.1}(i). Then, there exists $g_2 \in (0,\, g_1]$
such that $H$ has a unique ground state for $g\leq g_2$.
Moreover, for
 $$
  E  =\inf\spec(H)\, ,
 $$
the spectrum of $H$ fulfils
\begin{equation*}
 \spec(H) = \spec_{\mathrm{ac}}(H) = [E,\, \infty)\ ,
\end{equation*}
with $E \leq 0$.
\end{theorem}
%
%
\begin{proof}
The proof of Theorem~\ref{thm:main-1} is done in
\cite{BarbarouxGuillot2009}. The main ingredients of this proof
are the cutoff operators and the existence of a gap above the
ground state energy for these operators (see \eqref{eq:def-H-n-up}
and Proposition~\ref{prop:gap} below and
\cite[Proposition~3.5]{BarbarouxGuillot2009}).
Note that a more general proof of the existence of a ground state
can also be achieved by mimicking the proof given in
\cite{Barbarouxetal2004}.
\end{proof}

Let $b$ be the operator in $L^2(\Sigma_1)$ accounting for the
position of the neutrino
 $$
  b = i\nabla_{p_2}\ ,
 $$
and let
 $$
  \lbr = (1 + |b|^2)^\frac12\, .
 $$
Its second quantized version $\d\Gamma( \lbr)$ is self-adjoint in
$\gF_a (L^2(\Sigma_1))$. We thus define the ``total position''
operator $B$ for all neutrinos and antineutrinos by
\begin{equation}\nonumber%
\begin{split}
 & B_\ell = \1\otimes\1\otimes \d\Gamma(\lbr) \otimes \1
 + \1\otimes\1\otimes\1\otimes \d\Gamma(\lbr)
 \quad \mbox{ on } \gF_\ell, \\
 & B =   \left(B_1 \otimes \1  \otimes \1
        + \1  \otimes B_2 \otimes \1
        + \1  \otimes \1  \otimes B_3\right)\otimes(\1\otimes\1)
 \quad \mbox{ on } \gF\, .
\end{split}
\end{equation}

\begin{theorem}\label{thm:main-2}
Assume that the kernels $G^{(\alpha)}_{\ell,\epsilon,\epsilon'}$
satisfy Hypotheses~\ref{hypothesis:2.1} and
\ref{hypothesis:3.1}~(ii)-(iii). For any $\delta>0$ satisfying
$0<\delta<m_1$, there exists $g_\delta>0$ such that for $0 <g \leq
g_\delta$:
\begin{itemize}
\item[(i)] The spectrum of $H$ in $(E,\, E + m_1-\delta]$ is
    purely absolutely continuous.

\item[(ii)] For $s>1/2$, $\varphi\in\gF$, and $\psi\in\gF$,
    the limits
\begin{equation*}
 \lim_{\epsilon\to 0} (\varphi,\, \langle B \rangle^{-s}
 (H-\lambda\pm i\epsilon) \langle B \rangle^{-s} \psi)
\end{equation*}
exist uniformly for $\lambda$ in every compact subset of
$(E,\, E + m_1 - \delta)$.

\item[(iii)] For $s\in (1/2,\, 1)$ and $f\in
    C_0^\infty\left((E,\, E + m_1 - \delta)\right)$, we have
\begin{equation*}
 \left\|
  (B+1)^{-s}  \mathrm{e}^{-i t H} f(H)
  (B+1)^{-s}
 \right\|
 = \mathcal{O}\big(t^{-(s-1/2)}\big)\, .
\end{equation*}
\end{itemize}
\end{theorem}
The assertions (i), (ii), and (iii) of Theorem~\ref{thm:main-2} are
based on a limiting absorption principle stated in
section~\ref{S-LAP}, obtained by a positive commutator estimate,
called Mourre estimate (section~\ref{S-Mourre}), and a regularity
property of $H$ (section~\ref{S-regularity}).

The proof of Theorem~\ref{thm:main-2} is detailed in
section~\ref{S-proof}.

\begin{remark}\label{rem:reduction}
As a representative example of the general process described
above, we consider the decay \eqref{eq:ivb} of
the intermediate vector boson $W^-$ into an electron and an
electron antineutrino.
All Theorems stated in sections~\ref{S2} and
\ref{S:spectral-properties} will obviously remain true for this
simplified model, as well as for any other reduced model involving
only one species of leptons, i.e., for a fixed value of $\ell \in
\{ 1,\,2,\,3\}$, and with or without the inclusion of their
corresponding antiparticles ($\epsilon=\pm$ and $\epsilon'=\pm$).
Moreover, the proofs of these results, based on the theorems
stated in sections~\ref{S-gap}, \ref{S-Mourre}, \ref{S-regularity},
and \ref{S-LAP}, follow exactly the same arguments and estimates
in the general case as in the case of fixed $\ell$, $\epsilon$ and
$\epsilon'$.
For this reason, and for the sake of clarity, we shall fix
$\ell=1$, $\epsilon=+$ and $\epsilon'=-$ in the next sections, and
we shall adopt the following obvious notations
\begin{equation}\label{eq:reduced-notations}
 b^\sharp(\xi_1)=b_{1,+}^\sharp(\xi_1),\quad
 c^\sharp(\xi_2)=c_{1,-}^\sharp(\xi_2),\quad
 a^\sharp(\xi_3)=a_+^\sharp(\xi_3),\quad  
G^{(\alpha)}=G^{(\alpha)}_{\ell,\epsilon,\epsilon'}\, .
\end{equation}
\end{remark}

\begin{remark} 
Let us comment on two alternative approaches. The first one consists in
confining the interaction to a large box. This is actually exactly the first step in the procedure introduced by Glimm and Jaffe in their attempt to define a Hamiltonian in Fock space. In our case, this approach has been described in the introduction of Barbaroux and Guillot 
\cite{BarbarouxGuillot2009}. It is 
well-known that such a simple confinement is unfortunately not sufficient
to make the Hamiltonian well-defined in Fock space (see the reference
{\rm [16]} to the work of Glimm and Jaffe given in 
\cite{BarbarouxGuillot2009}). In order to achieve
this goal, one is obliged to introduce ultraviolet cutoffs in the
momentum variables. In this way, one then gets square integrable interaction kernels which is exactly the setup we are using in the present paper. The second approach consists in 
fibering out total momentum. Indeed, some time ago,
two of the authors tried to implement this approach but, unfortunately, the obstacles appearing in the estimates which were supposed to lead to
a well-defined setup in Fock space could not be surmounted and, hence,
the main theorem of the present paper could not be established. This
problem remains an open and difficult question, and we intend to come
back to its study in near future.
In conclusion, to the best of our knowledge, the conditions on the interaction kernels given in the present paper constitute the most general criterion for our main theorem to hold true.
\end{remark}

\begin{remark}
We  would like to underline the importance of our result
and to position it w.r.t. similar results on the absolute continuity of
the spectrum above the ground state. Fr\"ohlich, Griesemer, and Sigal have recently proved absolute continuity
of the spectrum between the ground state energy and the first threshold for a system of confined atoms interacting with a quantized electromagnetic field (see \cite{Frohlichetal2008}). In order to do so, they used the
Pauli-Fierz transformation (called the Power-Zienau-Woolley transformation
by the physicists) which is of great use for the study of the infrared problem. Unfortunately, this transformation does not have an analogue
for the quantum field theoretical models we are studying. Hence, 
although our paper is strongly inspired by \cite{Frohlichetal2008}, and the overall  scheme of \cite{Frohlichetal2008} and ours are similar, our
method yields a more general approach to this class of problems 
leading to an alternative proof of the results by the above-mentioned
authors. Very recently, Chen, Faupin, Fr\"ohlich, and Sigal \cite{Chenetal2009}
studied the same problem as ours in the case of a model describing an electron which
interacts with a quantized electromagnetic field. They showed absolute 
continuity of the spectrum of the corresponding Hamiltonian at fixed
total momentum by using the Feshbach-Schur map. We have verified that the same approach can be used for our case. But the proof we give
in our paper is much simpler than the one using the Feshbach-Schur map. However, the use of the Feshbach-Schur 
map seems not to be avoidable  in the work of the above-mentioned authors.
\end{remark}

\section{Spectral gap for cutoff operators}\label{S-gap}

\setcounter{equation}{0}

A key ingredient of the proof of Theorem~\ref{thm:main-1} and
Theorem~\ref{thm:main-2} is the study of cutoff operators
associated with infrared cutoff Hamiltonians with respect to the
momenta of the neutrinos.
The main result of this section is Proposition~\ref{prop:gap}
where we prove that the cutoff operators have a gap in their spectrum
above the ground state energy. This property was already derived
in \cite{BarbarouxGuillot2009} in the case of a sharp ultraviolet
cutoff. We show here that this result remains true in the present
case where no sharp ultraviolet cutoff assumption is made.
According to Remark~\ref{rem:reduction}, for the sake of clarity,
we will consider only the case of one species $\ell=1$ of leptons,
and pick $\epsilon=+$, and $\epsilon'=-$. We thus use the
notations \eqref{eq:reduced-notations}.

Let us first define the cutoff operators which are the
Hamiltonians with infrared cutoff with respect to the momenta of
the neutrinos. For that purpose, let $\chi_0(.)\in C^\infty(\R, [0,1])$ with $\chi_0=1$ on
$(-\infty, 1]$, and, for $\sigma>0$ and $p\in\R^3$, we set
\begin{equation}\label{def:chitilde2}
\begin{split}
 &  \chi_\sigma(p) = \chi_0(|p|/\sigma)\ , \\
 &  \tilde\chi^\sigma(p) = 1 - \chi_\sigma(p)\ .
\end{split}
\end{equation}
Moreover, the operator $H_{I, \sigma}$ is the interaction given by
\eqref{eq:HI-quadraticform-1}, \eqref{eq:HI-quadraticform-2}, and
\eqref{eq:HI-quadraticform-3} associated with the kernels
$\tilde{\chi}^\sigma(p_2) G^{(\alpha)}(\xi_1,\xi_2,\xi_3)$. We
then set
\begin{equation}\label{eq:def-Hsigma-down}
  H_\sigma = H_0 + g H_{I,\sigma}\ .
\end{equation}
Next, let
\begin{eqnarray*}
 & & \Sigma_{1,\sigma} = \Sigma_1 \cap \{(p_2, s_2);\
 |p_2|<\sigma\}\ , \quad
 \Sigma_1^\sig = \Sigma_1 \cap \{ (p_2,s_2);\ |p_2|\geq
 \sigma\}, \\
 & & \gF_{2,\sigma} = \gF_a(L^2(\Sigma_{1,\sigma})) \ ,\quad
 \gF_{2}^\sig = \gF_a(L^2(\Sigma_1^\sig))\ .
\end{eqnarray*}
The space $\gF_a(L^2(\Sigma_1))$ is the Fock space for the massive
leptons and $(\gF_{2,\sigma} \otimes \gF_{2}^\sig)$ is the Fock
space for the antineutrinos. 
Now, we set
\begin{eqnarray*}
 & & \gF_L^\sig   = \gF_a(L^2(\Sigma_1)) \otimes \gF_{2}^\sig\
 ,\ \mbox{and}\quad
 \gF_{L,\sigma} = \gF_{2,\sigma}\ ,
\end{eqnarray*}
and, we thus have
\begin{equation*}
 \gF_L \simeq \gF_L^\sig \otimes \gF_{L,\sigma}\ .
\end{equation*}
Moreover, with
\begin{equation}\nonumber
\begin{split}
 & \gF^\sigma = \gF_L^\sigma \otimes\gF_W\ , \quad\mbox{and}\quad
 \gF_\sigma = \gF_{L,\sigma}\, ,
 \\
\end{split}
\end{equation}
we can write
\begin{equation*}
 \gF \simeq  \gF^\sigma \otimes \gF_\sigma\ .
\end{equation*}
Next, we set
\begin{equation*}
\begin{split}
 & H_0^{(1)} =
 \int w^{(1)}(\xi_1)\,
 b^*(\xi_1) b(\xi_1)
 \d\xi_1\ , \\
 & H_0^{(2)} =
 \int w^{(2)}(\xi_2)\,
 c^*(\xi_2) c(\xi_2)
 \d\xi_2\ , \\
 & H_0^{(3)} =
 \int w^{(3)}(\xi_3)
 a^*(\xi_3) a(\xi_3)
 \d\xi_3\ , \\
\end{split}
\end{equation*}
and
\begin{equation}\label{eq:def-h02-up-and-down}
\begin{split}
 & H_0^{(2)\sigma} =
 \int_{|p_2|>\sigma} w^{(2)}(\xi_2)\,
 c^*(\xi_2) c(\xi_2)
 \d\xi_2\ , \\
 & H_{0,\sigma}^{(2)} =
 \int_{|p_2|\leq\sigma} w^{(2)}(\xi_2)\,
 c^*(\xi_2) c(\xi_2)
 \d\xi_2\ .
\end{split}
\end{equation}
Then, on $\gF^\sigma\otimes\gF_\sigma$, we have 
\begin{equation*}
  H_0^{(2)} = H_0^{(2)\sigma} \otimes \1_\sigma +
  \1^\sigma\otimes H_{0,\sigma}^{(2)}\ ,
\end{equation*}
where $\1^{\sigma}$ (resp. $\1_\sigma$) is the identity operator
on $\gF^{\sigma}$ (resp. $\gF_\sigma$).
Next, using the definitions
\begin{equation}\nonumber
 H^{\sigma} = H_\sigma|_{\gF^{\,\sigma}}\quad\mbox{and}\quad
 H_0^{\,\sigma} = H_0|_{\gF^\sigma}\ ,
\end{equation}
we get
\begin{equation}\nonumber
  H^{\sigma} = H_0^{(1)} + H_0^{(2)\,\sigma} + H_0^{(3)}+
  g H_{I,\sigma}\quad\mbox{on}\ \gF^{\,\sigma}\ ,
\end{equation}
and
\begin{equation}\label{eq:def-h-sigma-down}
  H_\sigma = H^{\sigma}\otimes\1_\sigma +
  \1^{\,\sigma}\otimes H_{0,\sigma}^{(2)}
  \quad\mbox{on}\ \gF^{\,\sigma}\otimes\gF_\sigma\ .
\end{equation}

Moreover, for $\delta\in\R$ with $  0< \delta< m_1$, 
we define the sequence $(\sigma_n)_{n\geq 0}$ by
\begin{equation}\label{eq:def-sigma-n}
\begin{split}
 & \sigma_0 = 2m_1 +1 \ ,\\
 & \sigma_1 = m_1-\frac{\delta}{2}\ , \\
 & \sigma_{n+1} = \gamma \sigma_n \mbox{ for } n\geq 1\ ,
\end{split}
\end{equation}
where
\begin{equation}\label{eq:def-gamma}
 \gamma = 1 - \frac{\delta}{2m_1 - \delta}\ .
\end{equation}
For $n\geq 0$, we then define the cutoff operators on
$\gF^n=\gF^{\sigma_n}$ by
\begin{equation}\label{eq:def-H-n-up}
H^n = H^{\sigma_n},\quad H_0^n = H_0^{\sigma_n},
\end{equation}
and we denote, for $n\geq 0$,
\begin{equation}\label{eq:def-E-n-up}
 E^n = \inf \spec(H^n)\, .
\end{equation}
Furthermore, we set
\begin{equation}\label{def:D-tilde-delta}
 \tdd(G) =
 \max\left\{ \frac{4(2m_1+1)\gamma}{2 m_1 - \delta},\ 2\right\}
 \kt(G)(2 m_1 \ct + \bt)\ ,
\end{equation}
where $\kt(G)$ is given by
Hypothesis~\ref{hypothesis:3.1}\textit{(iii-a)} and $\bt$ and
$\ct$ are defined for given $\eta>0$ and $\beta>0$ as in
\cite[(3.29)]{BarbarouxGuillot2009} by the  relations
\begin{equation}\label{eq:3.26}
\begin{split}
 & C_{\beta\eta} =
 \left( \frac{3}{m_W} \left(1 + \frac{1}{m_1{}^2}\right)
 + \frac{3\beta}{m_W m_1{}^2} +
 \frac{12\,\eta}{m_1{}^2}(1+\beta)\right)^\frac12\ ,\\
 & B_{\beta\eta} =
 \left( \frac{3}{m_W} \left(1 + \frac{1}{4\beta}\right)
 + 12 \left(\,\eta\left(1+ \frac{1}{4\beta}\right) + \frac{1}{4\eta}\,\right)
 \right)^\frac12\ ,
\end{split}
\end{equation}
\begin{equation*}
 \tilde{C}_{\beta\eta} = C_{\beta\eta}
 \left(1 + \frac{g_1 K(G) C_{\beta\eta}}{1 - g_1 K(G)
 C_{\beta\eta}}\right) ,
\end{equation*}
\begin{equation*}
 \tilde{B}_{\beta\eta} =
 \left(1 + \frac{g_1\, K(G) C_{\beta\eta}}{1-g_1\, K(G)\, C_{\beta\eta}}
 \left(2 + \frac{g_1 K(G) B_{\beta\eta} C_{\beta\eta}}{ 1 - g_1 K(G)
 C_{\beta\eta}}\right)\right)B_{\beta\eta}\ ,
\end{equation*}
and
\begin{equation}\label{eq:3.28}
 K(G) = \left(
 \sum_{\alpha=1,2}
 \big\| G^{(\alpha)}\big\|^2\
 \right)^\frac12 .
\end{equation}
Finally, let $g_\delta^{(1)}$ be such that
\begin{equation}\label{eq:def-gdeltaun}
 0 < g_\delta^{(1)} < \min\left\{1,\ g_1,\
 \frac{\gamma-\gamma^2}{3 \tdd(G)}\right\}\ .
\end{equation}
We then have
\begin{proposition}\label{prop:gap}
 Suppose that the kernels
 $G^{(\alpha)}$ satisfy
 Hypotheses~\ref{hypothesis:2.1} and \ref{hypothesis:3.1}(ii).
 Then, there exists $\tilde{g}_\delta^{(2)}>0$ with
 $\tilde{g}_\delta^{(2)} \leq g_\delta^{(1)}$ such that, for $g\leq
 \tilde{g}_\delta^{(2)}$, $E^n$ is a simple eigenvalue of $H^n$
 for $n\geq 1$, and $H^n$ does not have spectrum in the interval 
$(E^n,\, E^n+
 (1 - 3 g \tdd(G)/\gamma)\sigma_n)$.
\end{proposition}
%
%
%
\begin{proof} The proof of Proposition~\ref{prop:gap} is a slight
modification of the proof of
\cite[Proposition~3.5]{BarbarouxGuillot2009} which was based on
the method develop in \cite{Bachetal2006}. The only difference
to the proof of \cite[Proposition~3.5]{BarbarouxGuillot2009}
is that we have to deal with the absence of the sharp ultraviolet
cutoff.
To do so, we define, for $n\geq 0$, 
\begin{equation}\nonumber
\Sigma_n^{n+1} = \Sigma_1 \cap \{(p_2,\, s_2);\, \sigma_{n+1}\leq |p_2|
<\sigma_n\},
\end{equation}
and
\begin{equation}\nonumber
 \gF_n^{n+1} = \gF_a\left( L^2(\Sigma_n^{n+1})\right)\, ,
\end{equation}
and we get
\begin{equation}\nonumber
 \gF^{n+1} \simeq \gF^n\otimes \gF_n^{n+1}\, .
\end{equation}
Now, let $\Omega^n$ (respectively $\Omega_n^{n+1}$) be the vacuum
state in $\gF^n$ (respectively in $\gF_n^{n+1}$),  and set
\begin{equation}\nonumber
 H_{0,n}^{n+1} = \int_{\sigma_{n+1} \leq |p_2| <\sigma_n}
 w^{(2)}(\xi_2) c^*(\xi_2) c(\xi_2) \d \xi_2\, ,
\end{equation}
which is self-adjoint in $\gF_n^{n+1}$.
Moreover, we denote by $H_I^n$ and $H_{I,n}^{n+1}$ the operators
defined as the interaction $H_I$ given by
\eqref{eq:HI-quadraticform-1}-\eqref{eq:HI-quadraticform-3} (for
$\ell=1$, $\epsilon=+$ and $\epsilon'=-$), but associated,
respectively, with the kernels
\begin{equation}\nonumber
 \tilde{\chi}^{\sigma_n}(p_2) G^{(\alpha)}(\xi_1,\, \xi_2,\, \xi_3),
\end{equation}
and
\begin{equation}\nonumber
 \left(
 \tilde{\chi}^{\sigma_{n+1}}(p_2) - \tilde{\chi}^{\sigma_n}(p_2)
 \right)
 G^{(\alpha)}(\xi_1,\, \xi_2,\, \xi_3),
\end{equation}
where $\tilde{\chi}^{\sigma_n}$ and $\tilde{\chi}^{\sigma_{n+1}}$
are defined as in \eqref{def:chitilde2}.
We also consider
\begin{equation}\nonumber
\begin{split}
 & H_+^n = H^n -E^n\, ,\\
 & \tilde{H}_+^n = H_+^n \otimes \1_n^{n+1}
 + \1_n\otimes H_{0,n}^{n+1}\, ,
\end{split}
\end{equation}
which are self-adjoint operators in $\gF^n$ and $\gF^{n+1}$,
respectively.
Combining \eqref{eq:2.53} of Lemma~\ref{lem:5.5} with
\eqref{eq:3.26} and \eqref{eq:3.28}, we obtain, for $n\geq 0$ and
$\psi\in\cD(H_{0}^{n})\subset\gF^n$,
\begin{equation}\label{eq:A7}
 g\| H_I^n \psi\| \leq g K(G) \left( C_{\beta\eta} \|H_0\psi\| +
 B_{\beta\eta}\|\psi\|\right)\, .
\end{equation}
It then follows from \cite[Section~V,~Theorem~4.11]{Kato1966} that
\begin{equation}\nonumber
\begin{split}
   & H^n \geq -\frac{g K(G) B_{\beta\eta}}{1 - g_1 K(G)
   C_{\beta\eta}} \geq - \frac{ g_1 K(G) B_{\beta\eta}}
   {1 - g_1 K(G) C_{\beta\eta}}\, ,\\
   & E^n \geq - \frac{ g K(G) B_{\beta\eta}}
   { 1 - g_1K(G) C_{\beta\eta}} \, .
\end{split}
\end{equation}
Since we have
\begin{equation}\nonumber
   ( \Omega^n,\, H^n \Omega^n ) = 0 ,
\end{equation}
we get
\begin{equation}\label{eq:A.11}
   E^n\leq 0\, ,\quad\mbox{and}\quad
   |E^n| \leq \frac{g K(G) B_{\beta\eta}}{1 - g_1 K(G)
   C_{\beta\eta}}\, .
\end{equation}
Let now
\begin{equation}\label{A.12}
 K_n^{n+1}(G) = K(\1_{\sigma_{n+1}\leq |p_2|\leq 2\sigma_n}G)\, .
\end{equation}
Combining \eqref{eq:2.53} with \eqref{eq:3.26} and \eqref{A.12},
we obtain, for $n\geq 0$ and $\psi\in \cD(H_0^{n+1})\subset\gF^{n+1}$,
\begin{equation}\nonumber
 g \|H_{I,n}^{n+1} \psi\|
 \leq g K_n^{n+1}(G) \left(C_{\beta\eta} \|H_0^{n+1} \psi\| +
 B_{\beta\eta} \|\psi\|\right)\, .
\end{equation}
Moreover, we  have
\begin{equation}\label{eq:A.14}
 H_0^{n+1} \psi = \tilde{H}_+^n\psi + E^n\psi -
 g(H_I^n\otimes\1_n^{n+1})\psi\, ,
\end{equation}
and by \eqref{eq:A7},
\begin{equation}\label{eq:A.15}
  g\| (H_I^n\otimes \1_n^{n+1})\psi\|
  \leq g K(G)
  (C_{\beta\eta} \|H_0^{n+1}\psi\| + B_{\beta\eta}\|\psi\|)\, .
\end{equation}
In view of \eqref{eq:A.11} and \eqref{eq:A.14}, it follows from
\eqref{eq:A.15} that
\begin{equation}\nonumber
\begin{split}
 & g\left\| \left( H_I^n\otimes\1_n^{n+1}\right)\psi\right\| \\
 & \leq \frac{g K(G) C_{\beta\eta}}{1 - g_1 K(G) C_{\beta\eta}}
 \| \tilde{H}_+^n\psi\|
 + \frac{g K(G) B_{\beta\eta}}{1 - g_1 K(G) C_{\beta\eta}}
 \left( 1 +\frac{g K(G) B_{\beta\eta}}{1 - g_1 K(G)
 C_{\beta\eta}}\right)
 \| \psi\|\, .
\end{split}
\end{equation}
Therefore, we obtain
\begin{equation}\label{eq:A.17}
 g  \| H_{I,n}^{n+1} \psi\|
 \leq g K_n^{n+1}(G) \left( \ct \| \tilde{H}_+^n\psi\|
 + \bt \|\psi\|\right)\, .
\end{equation}
Moreover, due to Hypothesis~\ref{hypothesis:3.1}\textit{(ii)}, we have
\begin{equation}\label{eq:A.18}
 K_n^{n+1}(G) \leq 2\, \sigma_n \tilde{K}(G)\, .
\end{equation}
Now, recall that for $n\geq 0$, we have $\sigma_{n+1} < m_1$.
Therefore, by \eqref{eq:A.17} and \eqref{eq:A.18}, we get
\begin{equation}\nonumber
 g\, \|H_{I,n}^{n+1} \psi\| \leq g\,
 K_n^{n+1}(G)\, \big(\, \tilde C_{\beta\eta} \| (\tilde H_+^n +
 \sigma_{n+1})\psi\| + (\tilde C_{\beta\eta}\, m_1 + \tilde
 B_{\beta\eta}) \|\psi\|\, \big)\ ,
\end{equation}
and, for $\phi\in\gF$,
\begin{equation}\label{eq:B.28}
\begin{split}
  g\| H_{I,n}^{n+1} (\tilde H_+^n +
 \sigma_{n+1})^{-1}\phi\|
 & \leq g\, K_n^{n+1}(G)\, \left(\, \tilde C_{\beta\eta} +
 \frac{ m_1 \tilde C_{\beta\eta} + \tilde
 B_{\beta\eta}}{\sigma_{n+1}}\, \right) \|\phi\| \\
 & \leq \frac{g}{\gamma}\,
 \tilde K(G) (2m_1 \tilde C_{\beta\eta} + \tilde B_{\beta\eta})
 \|\phi\|\ .
\end{split}
\end{equation}
Thus, by \eqref{eq:B.28}, the operator $H_{I,n}^{n+1} (\tilde
H_+^n + \sigma_{n+1})^{-1}$ is bounded and
\begin{equation}\nonumber
 g \|H_{I,n}^{n+1} (\tilde
 H_+^n + \sigma_{n+1})^{-1} \| \leq g\frac{\tilde
 {D}_{\delta}(G)}{\gamma}\ ,
\end{equation}
where $\tilde{D}_{\delta}(G)$ is given by
\eqref{def:D-tilde-delta}.
This yields, for $\psi\in\cD(\tilde H_+^n)$,
\begin{equation}\nonumber
 g \| H_{I,n}^{n+1} \psi\| \leq g \frac{\tilde{D}_{\delta}(G)}{\gamma}
  \| (\tilde H_+^n + \sigma_{n+1})\psi\| \ .
\end{equation}
Hence, it follows from \cite[\S V, Theorems~4.11 and
4.12]{Kato1966} that
\begin{equation}\label{eq:B.22}
  g | (H_{I,n}^{n+1}\psi,\, \psi)| \leq g\frac{\tilde{D}_{\delta}(G)}
  {\gamma} (\, (\tilde H_+^n + \sigma_{n+1})\psi,\,\psi\,)\ .
\end{equation}
For $g_\delta^{(1)}$ given by \eqref{eq:def-gdeltaun}, let
$g_\delta^{(2)}>0$ be such that
\begin{equation}\nonumber
 g_\delta^{(2)} \frac{\tilde{D}_{\delta}(G)}{\gamma} <1
 \quad\mbox{and}\quad
 g_\delta^{(2)} \leq g_\delta^{(1)}\ .
\end{equation}
By \eqref{eq:B.22}, we then get, for $g\leq g_\delta^{(2)}$,
\begin{equation}\label{eq:B.23}
 H^{n+1} = \tilde H_+^n + E^n + g H_{I,n}^{n+1}
 \geq E^n - \frac{g\, \tilde{D}_{\delta}(G)}{\gamma}\, \sigma_{n+1}
 + \Big(1- \frac{g\, \tilde{D}_{\delta}(G)}{\gamma}\Big) \tilde H_+^n\ .
\end{equation}
Because $(1- g \tilde{D}_{\delta}(G) / \gamma) \tilde H_+^n \geq
0$, it follows from \eqref{eq:B.23} that, for $n\geq 0$,
\begin{equation}\label{eq:B.24}
 E^{n+1} \geq E^n - \frac{g\, \tilde{D}_{\delta}(G)}{\gamma}\,
 \sigma_{n+1} .
\end{equation}
Suppose now that $\psi^n\in\gF^n$ with $\|\psi^n\|=1$ satisfies, for
$\epsilon>0$,
\begin{equation}\label{eq:B.25}
 E^n \leq (\psi^n,\, H^n \psi^n) \leq E^n + \epsilon \ .
\end{equation}
Then, for
\begin{equation}\label{eq:B.26}
 \tilde\psi^{n+1} = \psi^n\otimes\Omega_n^{n+1} \in \gF^{n+1}\ ,
\end{equation}
we obtain
\begin{equation}\label{eq:B.27}
 E^{n+1} \leq (\tilde\psi^{n+1},\, H^{n+1} \tilde\psi^{n+1})
 \leq E^n + \epsilon + g(\tilde\psi^{n+1},\ H_{I,n}^{n+1}
 \,\tilde\psi^{n+1})\, .
\end{equation}
By \eqref{eq:B.22}, \eqref{eq:B.25}, \eqref{eq:B.26}, and
\eqref{eq:B.27} we get, for every $\epsilon>0$,
\begin{equation}\nonumber
 E^{n+1} \leq E^n + \epsilon\left(1 + \frac{g\,\tilde{D}_{\delta}(G)}{\gamma}\right)
 + \frac{g\, \tilde{D}_{\delta}(G)}{\gamma} \, \sigma_{n+1}\ ,
\end{equation}
where $g\leq g_\delta^{(2)}$. This yields
\begin{equation}\label{eq:A.37}
 E^{n+1} \leq E^n + \frac{g\, \tilde{D}_{\delta}(G)}{\gamma}\, \sigma_{n+1}\ ,
\end{equation}
and by \eqref{eq:B.24}, we obtain
\begin{equation}\nonumber
 |E^n - E^{n+1}| \leq \frac{g\, \tilde{D}_{\delta}(G)}{\gamma}\, \sigma_{n+1}\
 .
\end{equation}

Let us first check that, for $\sigma_0$ given by
\eqref{eq:def-sigma-n}, $E^0$ is a simple isolated
eigenvalue of $H^0$ with
\begin{equation}\label{eq:A.19.1}
\inf \spec (H^0)\setminus\{E^0\} \geq m_1\, .
\end{equation}
Since
\begin{equation}\nonumber
\begin{split}
  g\| H_I(\1_{\sigma_0\leq |p_2|} G) \psi\|
 &\leq g K(G) (\cb\| H_0^0 \psi\|
 + \bb\|\psi\|) \\
 & \leq g K(G) (\cb\| (H_0^0 +1)\psi\|
 + (\cb + \bb)\| \psi \|)\, ,
\end{split}
\end{equation}
we get
\begin{equation}\label{eq:A.19.3}
 g \| H_I(\1_{\sigma_0\leq |p_2|} G)\psi\|
 \leq
 g K(G) (2\cb + \bb) \| (H_0^0 +1)\psi\|,
\end{equation}
and
\begin{equation}\label{eq:A.19.4}
 g \left|
 (\psi,\, H_I(\1_{\sigma_0\leq |p_2|} G) \psi)\right|
 \leq
 g K(G) (2\cb + \bb) (\psi,\, (H_0^0+1)\psi)\, .
\end{equation}
Set now
\begin{equation}\label{eq:A.19.5}
 \mu_2 =
 \sup_{\substack{\phi\in\gF^0\\ \phi\neq0}}\ \
 \inf_{\substack{\psi\in\cD(H^0) \\ (\psi,\phi)=0\\ \|\psi\|=1}}
 (\psi,\, H^0\psi)\, .
\end{equation}
By \eqref{eq:A.19.4} and \eqref{eq:A.19.5}, we have, for
$\Omega^0$ being the vacuum state in $\gF^{0}$,
\begin{equation}\nonumber
 \mu_2 \geq \inf_{\substack{\psi\in\cD(H^0)\\
 (\psi,\Omega^0)=0\\ \|\psi\|=1}}
 (\psi,\, H^0\psi)
 \geq \sigma_0 - gK(G) (2\cb + \bb)(\sigma_0+1) \, .
\end{equation}
Using the definition
\begin{equation}\nonumber
 g_3 = \frac{1}{2 K(G) (2\cb +\bb)}\, ,
\end{equation}
we get, for $g\leq g_3$,
\begin{equation}\nonumber
 \mu_2 \geq \frac{\sigma_0 - 1}{2} \geq E^0 + m_1,
\end{equation}
since $\sigma_0 = 2m_1+1$ and $E^0\leq 0$. Therefore, by
the min-max principle, $E^0$ is a simple eigenvalue of
$H^0$ and \eqref{eq:A.19.1} holds true.
We now conclude the proof of  Proposition~\ref{prop:gap} by
induction in $n\in\N$. Suppose that $E^n$ is a simple isolated
eigenvalue of $H^n$ such that, for $n\geq 1$,
\begin{equation}\nonumber
  \inf \left( \spec(H_+^n)\setminus\{0\} \right)
  \geq \Big(1 -\frac{3 g \tilde{D}_{\delta}(G)}{\gamma}\Big)
  \sigma_n\, .
\end{equation}
Due to \eqref{eq:def-sigma-n}-\eqref{eq:def-gdeltaun}, we have,
for $0<g\leq g_\delta^{(1)}$ and $n\geq 1$,
\begin{equation}\nonumber
 0 < \sigma_{n+1} < \Big(1- \frac{3 g \tdd(G)}{\gamma}\Big)\sigma_n \, .
\end{equation}
Therefore, for $g\leq g_\delta^{(2)}$, $0$ is also a simple
isolated eigenvalue of $\tilde{H}_+^n$ such that
\begin{equation}\label{eq:A.47}
  \inf \left( \spec(\tilde H_+^n)\setminus\{0\} \right)
  \geq \sigma_{n+1}\ .
\end{equation}
We now prove that $E^{n+1}$ is a simple isolated eigenvalue
of $H^{n+1}$ such that
\begin{equation}\nonumber
 \inf \left( \spec(H_+^{n+1}) \setminus\{0\}\right)
 \geq \Big(1-\frac{3 g \tilde{D}_{\delta}(G)}{\gamma}\Big)\sigma_{n+1}\ .
\end{equation}
To this end, let
\begin{equation}\nonumber
 \lambda^{(n+1)} = \sup_{\substack{\psi\in\gF^{n+1}\\ \psi\neq 0}}
 \ \ \inf_{\substack{(\phi,\psi)=0\\ \phi\in\cD(H^{n+1})\\ \|\phi\|=1}}
 (\phi,\, H_+^{n+1}\phi)\ .
\end{equation}
By \eqref{eq:B.23} and \eqref{eq:A.37}, we obtain, in $\gF^{n+1}$,
\begin{equation}\label{eq:A.50}
\begin{split}
 H_+^{n+1} & \geq E^n - E^{n+1}-\frac{g \tilde{D}_{\delta}(G)}
 {\gamma}\sigma_{n+1}
 + \Big(1- \frac{g \tilde{D}_{\delta}(G)}{\gamma}\Big)\tilde H_+^{n}\\
 & \geq \Big(1-\frac{g \tilde{D}_{\delta}(G)}{\gamma}\Big) \tilde H_+^n -
  \frac{2 g \tilde{D}_{\delta}(G)}{\gamma} \sigma_{n+1}\ .
\end{split}
\end{equation}
By \eqref{eq:B.26}, $\tilde\psi^{n+1}$ is the unique ground state
of $\tilde H_+^n$ and by \eqref{eq:A.47} and \eqref{eq:A.50}, we
have, for $g\leq g_\delta^{(2)}$,
\begin{equation}\nonumber
\begin{split}
 \lambda^{(n+1)} & \geq \inf_{\substack{(\phi,\tilde\psi^{n+1})=0\\
 \phi\in\cD(H^{n+1})\\ \|\phi\|=1}} (\phi, H_+^{n+1}\phi) \\
  & \geq \Big(1-\frac{g\tilde{D}_{\delta}(G)}{\gamma}\Big)\sigma_{n+1} - \frac{2 g
  \tilde{D}_{\delta}(G)}{\gamma} \sigma_{n+1}
  = \Big(1- \frac{3 g \tilde{D}_{\delta}(G)}{\gamma}\Big)
  \sigma_{n+1} >0\ .
\end{split}
\end{equation}
This concludes the proof of Proposition~\ref{prop:gap}, if one
proves that for
\begin{equation}\nonumber
 \tilde{g}_\delta^{(2)} = \min \left\{g_\delta^{(2)},\, g_3\right\}\, ,
\end{equation}
the operator $H^1$ satisfies the gap condition
\begin{equation}\label{eq:gap-cond-1}
 \inf\left(\spec(H_+^1)\setminus\{0\}\right)\geq
 \Big(1-\frac{3g\tilde{D}_{\delta}(G)}{\gamma}\Big)\sigma_1\, .
\end{equation}
By noting that $0$ is a simple isolated eigenvalue of $\tilde
H_+^0$ with $\inf(\spec(\tilde H_+^0)\setminus\{0\} )
\geq\sigma_1$, we prove that $E^1$ is indeed an isolated simple
eigenvalue of $H^1$ such that \eqref{eq:gap-cond-1} holds, by
mimicking the proof given above for $H_+^{n+1}$.
\end{proof}

\section{Mourre inequality}\label{S-Mourre}

\setcounter{equation}{0}

Let us set
\begin{equation}\nonumber
 \tau = 1 - \frac{\delta}{2(2m_1 - \delta)}\ .
\end{equation}
According to \eqref{eq:def-gamma}, we have
\begin{equation}\label{eq:4.14}
 0 < \gamma < \tau <
 1\quad\mbox{and}\quad\frac{\tau-\gamma}{2}<\gamma\ .
\end{equation}
Moreover, let $\chi^{(\tau)}\in C^\infty(\R,\, [0,1])$ be such that
\begin{equation}\nonumber
 \chi^{(\tau)}(\lambda) = \left\{
  \begin{array}{ll}
   1 & \mbox{ for } \lambda\in(-\infty,\, \tau]\, , \\
   0 & \mbox{ for } \lambda\in[1,\, \infty) \, .
  \end{array}
  \right.
\end{equation}
With the definition \eqref{eq:def-sigma-n} of $(\sigma_n)_{n\geq
0}$, we set, for all $p_2\in\R^3$ and $n\geq 1$,
\begin{equation}\nonumber
 \chi_n^{(\tau)} (p_2) =
 \chi^{(\tau)}\left(\frac{|p_2|}{\sigma_n}\right),
\end{equation}
\begin{equation}\label{def:a-n-tau}
 a_n^{(\tau)} = \chi_n^{(\tau)}(p_2) \frac12
 \left(p_2\cdot i\nabla_{p_2}
 + i\nabla_{p_2}\cdot p_2\right)
 \chi_n^{(\tau)}(p_2)\ ,
\end{equation}
and
\begin{equation}\label{def:A-n-tau}
 \ant = \1\otimes\d \Gamma(a_n^{(\tau)})\otimes\1\, ,
\end{equation}
where $\d\Gamma(.)$ refers to the usual second quantization of one
particle operators. The operators $a_n^{(\tau)}$ and $\ant$ are
self-adjoint, and we also have
\begin{equation}
 a_n^{(\tau)} = \frac12
 \left(\chi_n^{(\tau)}(p_2){}^2 p_2\cdot
 i\nabla_{p_2} +
 i\nabla_{p_2} \cdot p_2 \, \chi_n^{(\tau)}(p_2){}^2\right) \ .
\end{equation}
Next, let $N$ be the smallest integer such that
\begin{equation}\label{eq:4.19}
 N\gamma \geq 1\ .
\end{equation}
Due to \eqref{eq:def-sigma-n}-\eqref{eq:def-gdeltaun}, we have,
for $0<g\leq g_\delta^{(1)}$,
\begin{equation}\label{eq:4.9}
 0 < \gamma < \Big(1- \frac{3 g \tdd(G)}{\gamma}\Big)\ .
\end{equation}
Therefore, according to \eqref{eq:4.9} and \eqref{eq:4.19}, we
have
\begin{equation}\label{eq:4.20}
 \gamma < \gamma+\frac{1}{N}\left(1 - \frac{3 g \tdd(G)}{\gamma}
 -\gamma\right) < 1 - \frac{3 g \tdd}{\gamma}\ ,
\end{equation}
and
\begin{equation}\label{eq:4.21}
 \frac{\gamma}{N} \leq \gamma - \frac{1}{N}\left(
 1 - \frac{3 g \tdd(G)}{\gamma} -\gamma\right) < \gamma\ .
\end{equation}
Let us next define
\begin{equation}\label{eq:def-eps-gamma}
 \epsilon_\gamma = \min\left\{\frac{1}{2N} \Big(1-\frac{3 g
 \tdd}{\gamma} -\gamma\Big),\, \frac{\tau-\gamma}{4}\right\}\ ,
\end{equation}
and choose $f\in C_0^\infty(\R)$ such that $0\leq f \leq 1$ and
\begin{equation}\label{eq:4.23}
f(\lambda)=
 \left\{
 \begin{array}{ll}
  1 & \mbox{ if } \lambda\in [(\gamma-\epsilon_\gamma)^2,\, \gamma
  + \epsilon_\gamma]\, , \\
  0 & \mbox{ if } \lambda > \gamma+ 2\epsilon_\gamma\, , \\
  0 & \mbox{ if } \lambda < (\gamma - 2\epsilon_\gamma)^2\, .
 \end{array}
 \right.
\end{equation}
Note that using the definition \eqref{eq:def-eps-gamma} of
$\epsilon_\gamma$, \eqref{eq:4.20}, \eqref{eq:4.14}, and
\eqref{eq:4.21}, we have, for $g\leq g_\delta^{(1)}$,
\begin{equation}\nonumber
 \gamma + 2\epsilon_\gamma < 1  - \frac{3 g
 \tdd}{\gamma}\, ,
\end{equation}
where $g_\delta^{(1)}$ is defined by \eqref{eq:def-gdeltaun}. 
Moreover, we also have $\gamma + 2\epsilon_\gamma <\tau$, and
\begin{equation}\label{eq:4.26}
 \gamma - \epsilon_\gamma >\frac{\gamma}{N}\ .
\end{equation}
Next, for $n\geq 1$, we define
\begin{equation}\label{eq:def-fn}
 f_n(\lambda) = f\left(\frac{\lambda}{\sigma_n}\right)\ ,
\end{equation}
and
\begin{equation}\label{eq:def2-Hn-down}
 H_n = H_{\sigma_n}\,,
 \quad
 E_n=\inf\spec(H_n)
 \quad\mbox{and }
 H_{0,n}^{(2)} = H_{0,\sigma_n}^{(2)}\ ,
\end{equation}
where we used the definitions \eqref{eq:def-h02-up-and-down} and
\eqref{eq:def-h-sigma-down} for $H_{0,\sigma_n}^{(2)}$ and
$H_{\sigma_n}$. Note that $E_n = E^n$, where $E^n$ is defined by
\eqref{eq:def-E-n-up}. Let $P^n$ denote the ground state
projection of $H^n$. It follows from Proposition~\ref{prop:gap}
that, for $n\geq 1$ and $g\leq \tilde{g}_\delta^{(2)}$,
\begin{equation}\label{eq:4.32}
 f_n(H_n - E_n) = P^n\otimes f_n(H_{0, n}^{(2)})\ .
\end{equation}
For $E=\inf\spec(H)$ being the ground state energy of $H$ defined
in Theorem~\ref{thm:main-1}, and any interval $\Delta$, let
$E_\Delta(H-E)$ be the spectral projection for the operator
$(H-E)$ onto $\Delta$. Consider, for $n\geq 1$,
\begin{equation}\label{eq:def-Delta-n}
 \Delta_n = [(\gamma - \epsilon_\gamma)^2\sigma_n,\,
 (\gamma+\epsilon_\gamma)\sigma_n]\, .
\end{equation}

We are now ready to state the Mourre inequality.

\begin{theorem}[Mourre inequality]\label{thm:Mourre}
Suppose that the kernels $G^{(\alpha)}$ satisfy
Hypotheses~\ref{hypothesis:2.1}, \ref{hypothesis:3.1}(ii), and
\ref{hypothesis:3.1}(iii.a).
Then, there exists $C_\delta>0$ and $\tilde{g}_\delta^{(3)}>0$ with
 $\tilde{g}_\delta^{(3)} < \tilde{g}_\delta^{(2)}$ such that,
for $g < \tilde{g}_\delta^{(3)}$ and $n\geq 1$,
\begin{equation}\label{eq:mourre}
 E_{\Delta_n}(H-E) \, [H,\, i\ant]\, E_{\Delta_n}(H-E)
 \geq C_\delta \frac{\gamma^2}{N^2} \,\sigma_n\, E_{\Delta_n}(H-E)\, .
\end{equation}
\end{theorem}

\begin{proof}
Let us define
\begin{equation}\label{eq:def-Dfinite0}
 \begin{split}
   \gD_1 = & \{ \psi\in \gF_a(L^2(\Sigma_1))\ |\ \psi^{(n)}\in
   C_0^\infty \mbox{ for all $n\in\N$}, \mbox{ and }
   \psi^{(n)}\!=\!0\mbox{ for almost all $n$}\}\, ,\\
   \gD_2 = & \gD_1\, , \\
   \gD_W = & \{ \psi\in \gF_W\ |\ \psi^{(n)}\in
   C_0^\infty \mbox{ for all $n\in\N$ }, \mbox{ and }
   \psi^{(n)}\!=\!0\mbox{ for almost all $n$}\}\, ,
 \end{split}
\end{equation}
and consider the algebraic tensor product
 \begin{equation}\nonumber
  \gD = \gD_1 \hat\otimes \gD_2 \hat\otimes \gD_W\ .
 \end{equation}
According to \cite[Lemma~28]{DerezinskiGerard1999} and
\cite[Theorem~13]{Derezinski2006} (see also
\cite[Proposition~2.11]{Ammari2004}), one easily shows that the
sesquilinear form defined on $\gD\times\gD$ by
\begin{equation}\nonumber
  (\varphi,\,\psi) \mapsto (H\varphi,\, i\ant\psi)
  - (\ant \varphi,\, i H \psi)\, ,
\end{equation}
is the one associated with the following symmetric operator
denoted by $[H,\, i\ant]$,
\begin{equation}\label{eq:def-commutator}
  [H,\, i\ant]\psi
  = \left(
  \d\Gamma ((\chi_n^{(\tau)})^2 w^{(2)})
  \ +\ g \, H_I(-i (a_n^{(\tau)} G))\, \right)\psi\, .
\end{equation}
Let us prove that $[H,\, i\ant]$ is continuous for the graph
topology of $H$. Combining \eqref{eq:2.53} of Lemma~\ref{lem:5.5}
with \eqref{eq:3.26} and \eqref{eq:3.28}, we get, for $g\leq g_1$,
$n\geq 1$, and for $\psi\in\gD$,
\begin{equation}\label{eq:5.2}
 g \| H_I (-i (a_n^{(\tau)} G))\psi\| \leq
 g K(- i a_n^{(\tau)}G) \left(C_{\beta\eta} \|H_0\psi\|
 + B_{\beta\eta} \|\psi\| \right)\, .
\end{equation}
It follows from Hypothesis~\ref{hypothesis:3.1}\textit{(iii-a)}
that there exists a constant $\tilde{C}(G)$ such that, for $n\geq
1$,
\begin{equation}\label{eq:5.3}
 K(-i(a_n^{(\tau)}G)) \leq \tilde{C}(G) \sigma_n\,  .
\end{equation}
Moreover, we have, for $g\leq g_1$,
\begin{equation}\label{eq:5.4}
 \|H_0 \psi \| \leq \|H\psi\| + g\, \| H_I(G) \psi \|
               \leq \|H\psi\| + g K(G) \left(C_{\beta\eta}
               \|H_0\psi\| + B_{\beta\eta} \|\psi\|\right)\, ,
\end{equation}
and, by definition of $g_1$,  
\begin{equation}\label{eq:5.5}
 g_1 K(G) C_{\beta\eta} <1\ .
\end{equation}
Using \eqref{eq:5.4} and \eqref{eq:5.5}, we get
\begin{equation}\label{eq:5.6}
 \| H_0 \psi \| \leq \frac{1}{1 - g_1 K(G) C_{\beta\eta}}
 \left(\|H\psi\| + g_1 K(G) B_{\beta\eta}\|\psi\|\right)\ .
\end{equation}
Therefore, for $\psi\in\gD$,
\begin{equation}\label{eq:5.7}
 \| \d \Gamma( (\chi_n^{(\tau)})^2 w^{(2)})\psi\|
 \leq \|H_0\psi\|\leq
  \frac{1}{1 - g_1 K(G) C_{\beta\eta}}
 \left(\|H\psi\| + g_1 K(G) B_{\beta\eta}\|\psi\|\right)\ .
\end{equation}
By \eqref{eq:5.2}, \eqref{eq:5.3}, and \eqref{eq:5.6}, we get, for
$g\leq g_1$, $n\geq 1$ and $\psi\in\gD$,
\begin{equation}\label{eq:5.8}
\begin{split}
 & g \| H_I(-i(a_n^{(\tau)} G))\psi\| \\
 & \leq g \tilde{C}(G) \sigma_n
 \left( \frac{C_{\beta\eta}}{1 - g_1 K(G) C_{\beta\eta}} \|H\psi\|
 + \Big( \frac{g_1 K(G) C_{\beta\eta}}{1 - g_1 K(G) C_{\beta\eta}} +1\Big)
 B_{\beta\eta} \|\psi\| \right)\, .
\end{split}
\end{equation}
Since $\gD$ is a core for $H$, 
\eqref{eq:5.7} and \eqref{eq:5.8} are fulfilled for
$\psi\in\cD(H)$. Moreover, it follows from
\cite[Proposition~3.6(iii)]{BarbarouxGuillot2009} that $H$ is of
class $C^1(\ant)$ (see \cite[Theorem~6.3.4]{Amreinetal1996} and
condition (M') in \cite{GeorgescuGerard}) for $g\leq g_1$ and
$n\geq 1$. Therefore, \eqref{eq:def-commutator} holds for
$\psi\in\cD(H)$.

Recall from \eqref{eq:def-fn} that $f_n(\lambda)=
f(\lambda/\sigma_n)$, where $f$ is given by \eqref{eq:4.23}. Let
$\tilde{f}(.)$ be an almost analytic extension of $f(.)$
satisfying
\begin{equation}\label{eq:5.11}
 \left| \frac{\partial \tilde{f}}{\partial\bar z}(x+iy)\right|
 \leq C\, y^2\, .
\end{equation}
Note that
\begin{equation}\label{eq:5.12}
 \tilde{f}(x+iy) \in C_0^\infty(\R^2)\ ,
\end{equation}
and that
\begin{equation}\nonumber
 f(s) = \int \frac{\d\tilde{f}(z)}{z-s},\
 \quad
 \d\tilde{f}(z) = -\frac{1}{\pi} \frac{\partial \tilde{f}}{\partial
 \bar z} \d x \d y\ .
\end{equation}
It follows from \eqref{eq:4.32} that, for $g\leq
\tilde{g}_\delta^{(2)}$,
\begin{equation}\nonumber
\begin{split}
 \| \d\Gamma\big( (\chi_n^{(\tau)})^2 w^{(2)}\big) f_n(H_n -
 E_n)\|
 & = \| P^n\otimes \d\Gamma\big( (\chi_n^{(\tau)})^2 w^{(2)}\big)
 f_n(H_{0,n}^{(2)}) \| \\
 & \leq \| H_{0,n}^{(2)} f_n(H_{0,n}^{(2)})\|\ .
\end{split}
\end{equation}
Therefore, there exists $C_1^f >0$, depending on $f$, such that
for $g\leq \tilde{g}_\delta^{(2)}$,
\begin{equation}\label{eq:5.16}
 \left\| \d\Gamma\big( (\chi_n^{(\tau)})^2 w^{(2)}\big)
 f_n(H_n - E_n)\right\| \leq C_1^f\, \sigma_n\ .
\end{equation}
Recall (see \eqref{eq:def-Hsigma-down} and
\eqref{eq:def2-Hn-down}) that $H_n = H_0 + g H_{I,n}$, where $H_{I,
n}=H_{I,\sigma_n}$ is the interaction given by
\eqref{eq:HI-quadraticform-1}, \eqref{eq:HI-quadraticform-2}, and
\eqref{eq:HI-quadraticform-3} associated with the kernels
$\tilde{\chi}^{\sigma_n}(p_2) G^{(\alpha)}(\xi_1,\xi_2,\xi_3)$.
Now, in \eqref{eq:A.11}, it is stated
\begin{equation}\label{eq:5.18}
 |E^n| \leq \frac{g\, K(G)\, B_{\beta\eta}}
 {1 - g_1 K(G)\, C_{\beta\eta}} \ .
\end{equation}
Moreover, for $z\in\mathrm{supp}(\tilde{f})$, we have
\begin{equation}\label{eq:5.19}
\begin{split}
  (H_0 + 1) (H_n - E_n - z\sigma_n)^{-1}
 & = 1 + (E_n + z\sigma_n) (H_n - E_n - z\sigma_n)^{-1} \\
 & - g H_{I,n} (H_n - E_n - z\sigma_n)^{-1}
  + (H_n - E_n - z\sigma_n)^{-1} .
\end{split}
\end{equation}
Mimicking the proof of \eqref{eq:5.6} and \eqref{eq:5.8} and using
\eqref{eq:5.18}, we get for $g\leq g_1$,
\begin{equation}\label{eq:5.20}
\begin{split}
 &  g \| H_{I,n} (H_n - E_n - z\sigma_n)^{-1}\|
  \\
 & \leq \frac{g_1 K(G) \cb}{ 1 - g_1 K(G)\cb}
 \left(
  1 + \left( \frac{g_1 K(G) \bb}{1-g_1 K(G) \cb}
  + |z|\sigma_n + \frac{g_1 K(G) \bb}{1 - g_1 K(G)\cb}
  \right) \frac{1}{|\Im z|\sigma_n}
  \right) \\
  & \ \ \ + \frac{g_1 K(G) \bb}{|\Im z|\sigma_n}\, .
\end{split}
\end{equation}
It follows from \eqref{eq:5.18}, \eqref{eq:5.19}, and
\eqref{eq:5.20} that there exists $\tilde{C}_2(G)>0$ such that,
for $g\leq g_1$ and $n\geq 1$,
\begin{equation}\label{eq:5.21}
 \| (H_0 + 1) (H_n - E_n - z \sigma_n)^{-1} \|
 \leq
 \tilde{C}_2(G)
 \ \frac{1 + |z| \sigma_n}{| \Im z |
  \sigma_n}
 \, .
\end{equation}
Mimicking the proof of \eqref{eq:5.21}, we show that there exists
$\tilde{C}_3(G)>0$ such that, for $g\leq g_1$ and $n\geq 1$,
\begin{equation}\label{eq:5.22}
 \| (H_0 + 1) (H - E - z\sigma_n)^{-1}\|
 \leq
 \tilde{C}_3(G) \ \frac{ 1 + |z| \sigma_n}{ | \Im z | \sigma_n }\, .
\end{equation}
Furthermore, we have
\begin{equation}\label{eq:5.23}
\begin{split}
 & g H_I (-i (a_n^{(\tau)} G)) f_n(H_n - E_n) \\
 & =
 - \sigma_n \int \d\tilde{f}(z) H_I (-i (a_n^{(\tau)} G))
 (H_0 + 1)^{-1} (H_0 + 1) (H_n - E_n - z\sigma_n)^{-1}\, .
\end{split}
\end{equation}
By \eqref{eq:5.4}, \eqref{eq:5.8}, \eqref{eq:5.21}, and
\eqref{eq:5.23}, there exists $\tilde{C}_4^f(G) >0$ depending on
$f$, such that for $g\leq g_1$,
\begin{equation}\label{eq:5.24}
 g \left\|
 H_I( -i(a_n^{(\tau)} G)) f_n(H_n - E_n) \right\|
 \leq g\,
 \tilde{C}_4 ^f(G)\, \sigma_n\ .
\end{equation}
Similarly, by \eqref{eq:5.22}, we easily show that there exists
$\tilde{C}_5^f(G)>0$, depending on $f$, such that, for $g\leq g_1$,
\begin{equation}\label{eq:5.25}
 g \left\| H_I(-i (a_n^{(\tau)} G )) f_n(H-E) \right\|
 \leq g\, \tilde{C}_5 ^f(G)\, \sigma_n\ .
\end{equation}
By \eqref{eq:4.32}, we have, for $g\leq \tilde{g}_\delta^{(2)}$,
\begin{equation}\label{eq:5.54}
 f_n(H_n - E_n) \d \Gamma( (\cnt)^2 w^{(2)})
 f_n(H_n - E_n) =
 P^n \otimes f_n(H_{0,n}^{(2)})
 \d\Gamma ( (\cnt)^2 w^{(2)}) f_n(H_{0,n}^{(2)}) .
\end{equation}
Since $\cnt(\lambda) = 1$ if $\lambda \leq (\gamma +
2\epsilon_\gamma) \sigma_n$, we have
\begin{equation}\label{eq:5.55}
 f_n(H^{(2)}_{0,n})\, \d\Gamma\! \left( (\cnt)^2 w^{(2)}\right)
  f_n(H_{0,n}^{(2)})
 =
 f_n(H_{0,n}^{(2)})\, H_{0,n}^{(2)}\, f_n(H_{0,n}^{(2)})\ .
\end{equation}
Now, by \eqref{eq:4.23}, \eqref{eq:4.26}, \eqref{eq:5.54}, and
\eqref{eq:5.55}, we obtain, with $g\leq \tilde{g}_\delta^{(2)}$
and $n\geq 1$,
\begin{equation}\label{eq:5.56}
\begin{split}
 f_n(H_n - E_n)\, \d\Gamma\!\left((\cnt)^2 w^{(2)}\right) f_n(H_n
 - E_n)
  & \geq (\inf \mathrm{supp}(f_n)) f_n(H_n - E_n)^2 \\
  & \geq \frac{\gamma^2}{N^2} \sigma_n f_n(H_n - E_n)^2\, .
\end{split}
\end{equation}
Note that
\begin{equation}\label{eq:5.56-bis}
 \| f_n(H_n - E_n) \| = \| f_n(H-E) \|
 = \sup_\lambda | f_n(\lambda)| =1\, .
\end{equation}
By \eqref{eq:5.24} and \eqref{eq:5.56-bis} we get, for $g\leq
g_1$,
\begin{equation}\label{eq:5.57}
 f_n(H_n - E_n)
 g H_I(- i (a_n^{(\tau)})G) f_n(H_n - E_n) \geq
 - g \tilde{C}_4^f(G) \sigma_n\, .
\end{equation}
Thus, using
\eqref{eq:5.56} and \eqref{eq:5.57}, we get, for $g\leq \tilde{g}_\delta^{(2)}$, 
\begin{equation}\label{eq:5.58}
 f_n(H_n - E_n) [H,\, i \ant] f_n(H_n - E_n) \geq
 \frac{\gamma^2}{N^2} \sigma_n f_n(H_n - E_n)^2 - g \tilde{C}_4^f
 (G) \sigma_n\, .
\end{equation}
Next, let us make the decomposition
\begin{equation}\label{eq:5.59}
\begin{split}
 & f_n(H-E) [H,\, i\ant] f_n(H-E) \\
 & =
 f_n(H_n - E_n) [H,\, i\ant] f_n(H_n - E_n) \\
 &
 + \left( f_n(H-E) - f_n(H_n - E_n)\right)
 [H,\, i\ant] f_n(H_n - E_n) \\
 &
 + f_n(H-E) [H,\, i\ant] \left( f_n(H-E) - f_n(H_n - E_n)\right)\,
 .
\end{split}
\end{equation}
Using \eqref{eq:5.16} and Lemma~\ref{lem:5.3}, we get, for $g\leq
\tilde{g}_\delta^{(2)}$,
\begin{equation}\label{eq:5.60}
\left( f_n(H-E) - f_n(H_n - E_n) \right) \d\Gamma \left( (\cnt)^2
w^{(2)} \right) f_n(H_n - E_n) \geq - g C_1^f \tilde{C}_6^f(G)
\sigma_n\, .
\end{equation}
By \eqref{eq:5.24} and Lemma~\ref{lem:5.3}, we obtain, for $g\leq
g_2$,
\begin{equation}\label{eq:5.61}
g \left( f_n(H-E) - f_n(H_n - E_n)\right) H_I(-i (a_n^{(\tau)}G))
f_n(H_n - E_n) \geq - g g_2 \tilde{C}_4^f(G) \tilde{C}_6^f(G)
\sigma_n\, .
\end{equation}
Thus, it follows from \eqref{eq:5.60} and \eqref{eq:5.61} that,
for $g\leq \inf(g_2,\, \tilde{g}_\delta^{(2)})$,
\begin{equation}\label{eq:5.62}
\begin{split}
 & \left( f_n(H-E) - f_n(H_n - E_n)\right) [H,\, i\ant] f_n(H_n -
 E_n) \\
 & \geq
 - g \tilde{C}_6^f(G)\left(C_1^f +
 g_2\tilde{C}_4^f(G)\right)\sigma_n\, .
\end{split}
\end{equation}
Similarly, by Lemma~\ref{lem:5.4} and \eqref{eq:5.56}, we obtain,
for $g\leq \inf(g_2,\, \tilde{g}_\delta^{(2)})$,
\begin{equation}\label{eq:5.63}
 f_n(H-E) \d\Gamma\left( (\cnt)^2 w^{(2)}\right)
 \left(f_n(H-E) - f_n(H_n - E_n)\right) \geq - g
 \tilde{C}_7^f(G)\sigma_n\, .
\end{equation}
Moreover by \eqref{eq:5.25} and Lemma~\ref{lem:5.3} we get, for
$g\leq g_2$,
\begin{equation}\label{eq:5.64}
 g f_n(H-E) H_I(-i (a_n^{(\tau)} G))
 \left( f_n(H-E) - f_n(H_n-E_n)\right) \geq
 - g g_1 \tilde{C}_5^f(G) \tilde{C}_6^f(G) \sigma_n\, .
\end{equation}
Thus, it follows from \eqref{eq:5.63} and \eqref{eq:5.64} that,
for $g\leq \inf (g_2,\, \tilde{g}_\delta^{(2)})$, 
\begin{equation}\label{eq:5.65}
\begin{split}
 & f_n(H-E) [H,\, i\ant] \left( f_n(H-E) - f_n(H_n -
 E_n)\right) \\
 & \geq -g \left( \tilde{C}_7^f(G) + g_1 \tilde{C}_5^f(G)
 \tilde{C}_6^f(G)\right) \sigma_n\, .
\end{split}
\end{equation}
Furthermore, by Lemma~\ref{lem:5.3} and \eqref{eq:5.56-bis},
 we easily get, for $g\leq g_2$, 
\begin{equation}\label{eq:5.66}
\begin{split}
 f_n(H_n - E_n)^2 = & f_n(H-E)^2 + \left( f_n(H_n -E_n)
 - f_n(H-E)\right)^2 \\
 & + f_n(H-E) \left( f_n(H_n - E_n) - f_n(H-E)\right) \\
 & + \left(f_n(H_n - E_n) - f_n(H-E)\right) f_n(H-E) \\
 & \geq f_n(H-E)^2 - g\tilde{C}_6^f(G) (g_2 \tilde{C}_6^f(G) +2)\,
 .
\end{split}
\end{equation}
It then follows from \eqref{eq:5.58} and \eqref{eq:5.66} that,
for $g\leq \inf(g_2,\, \tilde{g}_\delta^{(2)} )$,
\begin{equation}\label{eq:5.67}
\begin{split}
 & f_n( H_n - E_n ) [H,\, i\ant ] f_n (H_n - E_n) \\
 & \geq
 \frac{\gamma^2}{N^2} \sigma_n f_n(H-E)^2
 - g\sigma_n \left(\tilde{C}_4^f(G)
 + \frac{\gamma^2}{N^2} \tilde{C}_6^f(G)
 \left( g_2 \tilde{C}_6^f(G) + 2  \right)
 \right)\, .
\end{split}
\end{equation}
Combining \eqref{eq:5.59} with \eqref{eq:5.62}, \eqref{eq:5.65},
and \eqref{eq:5.67}, we obtain, for $g\leq\inf(g_2,\,
\tilde{g}_\delta^{(2)})$,
\begin{equation}\label{eq:5.68}
 f_n(H -E) [H,\, i\ant] f_n(H-E)
 \geq \frac{\gamma^2}{N^2} \sigma_n f_n(H-E)^2 - g\sigma_n
 \tilde{C}_\delta\, ,
\end{equation}
with $\tilde{C}_\delta = \tilde{C}_6^f(G) (C_1^f + g_1
\tilde{C}_4^f(G)) + \tilde{C}_7^f(G) + g_1 \tilde{C_5}^f(G)
\tilde{C}_6^f(G) + \tilde{C}_4^f(G) +
\gamma^2/N^2 \tilde{C}_6^f(G) 
\linebreak(g_1\tilde{C}_6^f(G) + 2)$. Multiplying both sides of 
\eqref{eq:5.68}
with $E_{\Delta_n}(H-E)$, we  get
\begin{equation}\nonumber
 E_{\Delta_n} (H-E) [H,\, i\ant] E_{\Delta_n} (H-E)
 \geq \left( \frac{\gamma^2}{N^2} - g\tilde{C}_\delta\right)
 \sigma_n E_{\Delta_n}(H-E)\, .
\end{equation}
Picking a constant $\tilde{g}_\delta^{(3)}$ such that
\begin{equation}\label{eq:def-gtildedeltatrois}
 \tilde{g}_\delta^{(3)} < \min\Big\{ g_2,\, \tilde{g}_\delta^{(2)},
 \frac{\gamma^2}{N^2} \frac{1}{\tilde{C}_\delta}\Big\}\, ,
\end{equation}
Theorem~\ref{thm:Mourre} is proved, for $g\leq
\tilde{g}_\delta^{(3)}$ and $n\geq 1$, with $C_\delta =
\gamma^2(1 - N^2
 \tilde{C}_\delta \tilde{g}_\delta^{(3)}/\gamma^2)/N^2$.
\end{proof}

\section{$C^2(\ant)$-regularity}\label{S-regularity}

\setcounter{equation}{0}

\begin{theorem}\label{thm:regularity}
Suppose that the kernels $G^{(\alpha)}$ satisfy
Hypotheses~\ref{hypothesis:2.1} and \ref{hypothesis:3.1}(iii). Then, $H$ is locally of class
$C^2(\ant)$ in $(-\infty,\, m_1-\frac{\delta}{2})$ for every
$n\geq 1$.
\end{theorem}
%
%
%
\begin{proof}
The proof is achieved by substituting $\ant$ for $A_{\sigma}$ in
the proof of Theorem~3.7 in \cite{BarbarouxGuillot2009}.
\end{proof}
\begin{remark}
It is likely that the operator $H$ is of class $C^2(\ant)$, i.e.,
not only locally.
\end{remark}

\section{Limiting Absorption Principle}\label{S-LAP}

\setcounter{equation}{0}

For $\ant$ defined by \eqref{def:A-n-tau}, we set
\begin{equation}\nonumber
  \langle \ant\rangle = (1 + (\ant)^2)^\frac12\ .
\end{equation}
Recall that $[\sigma_{n+2},\, \sigma_{n+1}] \subset
 \Delta_n = [\, (\gamma-\epsilon_\gamma)^2\sigma_n,\,
 (\gamma + \epsilon_\gamma)\sigma_n]$ for $n\geq 1$.

\begin{theorem}[Limiting Absorption Principle]\label{thm:LAP}
Suppose that the kernels $G^{(\alpha)}$ satisfy
Hypotheses~\ref{hypothesis:2.1}, \ref{hypothesis:3.1}~(ii), and
\ref{hypothesis:3.1}~(iii). Then, for any $\delta>0$ satisfying
$0<\delta<m_1$, there exists $g_\delta>0$ such that, for $0<g \leq
g_\delta$, for $s>1/2$, $\varphi$, $\psi\in\gF$ and for $n\geq 1$,
the limits
\begin{equation}\nonumber
 \lim_{\epsilon\to 0} (\varphi,\, \langle\ant\rangle^{-s}
 (H-\lambda\pm i \epsilon) \langle\ant\rangle^{-s} \psi)
\end{equation}
exist uniformly for $\lambda\in\Delta_n$.
Moreover, for $1/2<s<1$, the map
\begin{equation}\nonumber
  \lambda \mapsto \langle\ant\rangle^{-s} (H-\lambda\pm i0)^{-1} \langle\ant\rangle^{-s}
\end{equation}
is H\"older continuous of degree $s-1/2$ in $\Delta_n$.
\end{theorem}
\begin{proof}
Theorem~\ref{thm:LAP} follows from the $C^2(\ant)$-regularity in
Theorem~\ref{thm:regularity} and the Mourre inequality in
Theorem~\ref{thm:Mourre} with $g_\delta= \tilde{g}_\delta^{(3)}$
(defined by \eqref{eq:def-gtildedeltatrois}), according to
Theorems~0.1 and 0.2 in \cite{Sahbani1997} (see also
\cite{GoleniaJecko2007}, \cite{Gerard2008}, and
\cite{Frohlichetal2008}).
\end{proof}

\section{Proof of Theorem~\ref{thm:main-2}}\label{S-proof}

\setcounter{equation}{0}

$\bullet$ We first prove (i) of Theorem~\ref{thm:main-2}.
According to \eqref{eq:def-sigma-n} we have
 $$
  [\sigma_{n+2},\, \sigma_{n+1}]
  \subset [(\gamma-\epsilon_\gamma)^2\sigma_n,\,
  (\gamma+\epsilon_\gamma)\sigma_n] = \Delta_n\, ,
 $$
thus $\bigcup_{n} \Delta_n$ is a covering by open sets of any
compact subset of $(\inf\spec(H),\, m_1-\delta)$. Therefore,
\cite[Theorem~0.1 and Theorem~0.2]{Sahbani1997} together with the
Mourre inequality \eqref{eq:mourre} in Theorem~\ref{thm:Mourre}
and the local $C^2(\ant)$ regularity in
Theorem~\ref{thm:regularity} imply that (i) of
Theorem~\ref{thm:main-2} holds true.

\medskip

$\bullet$ For the proof of (ii) of Theorem~\ref{thm:main-2}, let
us first note that since $\bigcup_{n} \Delta_n$ is a covering by
intervals of $(E,\,E+ m_1-\delta)$, using subadditivity,
it suffices to prove the result for any $n\geq1$ and $f\in
C_0^\infty(\Delta_n)$.

For $a_n^{(\tau)} = \chi_n^{(\tau)}(p_2) \frac12 \left(p_2\cdot
i\nabla_{p_2} + i\nabla_{p_2}\cdot p_2\right)\chi_n^{(\tau)}(p_2)$ as given by \eqref{def:a-n-tau}, and $b = i\nabla_{p_2}$, we
have, for all $\varphi\in\cD(b)$,
\begin{equation}\nonumber
\begin{split}
 \| a_n^{(\tau)} \varphi \| & =
 \| \chi_n^{(\tau)}(p_2) \frac12
 \left(p_2\cdot i\nabla_{p_2}
 + i\nabla_{p_2}\cdot p_2\right)
 \chi_n^{(\tau)}(p_2) \varphi\| \\
 & \leq \frac12( \|\cnt(p_2)\, p_2\|\ + \|p_2\cnt(p_2)\|) \|i\nabla_{p_2}\varphi\|
 + \frac12 \| i\nabla_{p_2}\, p_2 \cnt\|\, \|\varphi\|\, ,
\end{split}
\end{equation}
 Therefore, there exists $c_n>1$ such
that
\begin{equation}\nonumber
 |a_n^{(\tau)}|^2 \leq c_n \langle b \rangle^2\, .
\end{equation}
Since $\langle b \rangle$ is a nonnegative operator,
\cite[Proposition~3.4 ii)]{Georgescuetal2004} implies
\begin{equation}\nonumber
 \d\Gamma( a_n^{(\tau)} )^2 \leq c_n \d\Gamma( \langle b \rangle )^2\, ,
\end{equation}
and thus
\begin{equation}\nonumber
 (\ant)^2 \leq c_n B^2\ .
\end{equation}
This implies
\begin{equation}\label{eq:hi-0}
 \| (B+1)^{-1} \langle \ant\rangle\| <\infty\quad\mbox{and}\quad
 \| \langle \ant\rangle (B+1)^{-1}\| <\infty\, .
\end{equation}
The map
 $$
  F(z)=\mathrm{e}^{-z\ln (B+1)} \mathrm{e}^{z\ln \langle\ant\rangle}\phi
 $$
is analytic on the strip $S=\{ z\in\C\ |\ 0 < \Re\, z < 1\}$ for
all $\phi\in\cD(B)\subset\cD(\langle A_n^{(\tau)} \rangle)$. For
$\Re\, z =0$, the operator $F(z)$ is bounded by $\|\phi\|$ and,
for $\Re\, z=1$, according to \eqref{eq:hi-0}, $F(z)$ is bounded
by $\| (B+1)^{-1} \langle \ant\rangle\|\,\|\phi\|$. Therefore, due
to Hadamard's three-line theorem, $F(z)$ is a bounded operator on
the strip $S$. In particular, for all $s\in (0,1)$, we obtain
\begin{equation}\label{eq:hi-1}
 \| (B+1)^{-s} \langle \ant\rangle^s\| <\infty\quad\mbox{and}\quad
 \| \langle \ant\rangle^s (B+1)^{-s}\| <\infty\, .
\end{equation}
Using \eqref{eq:hi-1}, we can write
\begin{equation*}
\begin{split}
 & (\varphi,\, \langle B+1 \rangle^{-s} (H-\lambda\pm\epsilon)^{-1}
 \langle B+1 \rangle^{-s} \psi ) \\
 & ( \langle \ant \rangle^{s} \langle B+1 \rangle^{-s} \varphi
 ,\, \langle \ant\rangle^{-s} (H-\lambda\pm\epsilon)^{-1} \langle\ant\rangle^{-s}
 \langle \ant\rangle^s \langle B+1 \rangle^{-s} \psi)\, .
\end{split}
\end{equation*}
We thus conclude the proof of Theorem~\ref{thm:main-2}~(ii) by
using Theorem~\ref{thm:LAP}.

\medskip

$\bullet$ We finally prove (iii) of Theorem~\ref{thm:main-2}. For
that sake, we first need to establish the following lemma.
\begin{lemma}\label{lem:8.1}
Suppose that $s\in(1/2,\, 1)$ and that for some $n$, $f\in
C_0^\infty(\Delta_n)$. Then,
\begin{equation}\nonumber
 \left\|
  \langle \ant\rangle^{-s} \mathrm{e}^{-i t H} f(H) \langle\ant\rangle^{-s}
 \right\|
 =
 \mathcal{O}\left(t^{-(s-\frac12)}\right)\, .
\end{equation}
\end{lemma}
\begin{proof}
The proof is the same as the one  in
\cite[Theorem~25]{Frohlichetal2008} for the Pauli-Fierz model of
non-relativistic QED. It makes use of the local H\"older
continuity stated in Theorem~\ref{thm:LAP}.
\end{proof}

We now prove (iii) of Theorem~\ref{thm:main-2} by using
\eqref{eq:hi-1}, Lemma~\ref{lem:8.1}, and writing
\begin{equation*}
\begin{split}
  & \| (B+1)^{-s} \mathrm{e}^{i t H} f(H) (B+1)^{-s} \| \\
  & \leq \| (B+1)^{-s} \langle \ant\rangle^s\|\,
  \| \langle\ant\rangle^{-s} \mathrm{e}^{i t H} f(H) \langle\ant\rangle^{-s} \|\,
  \| \langle \ant\rangle^s (B+1)^{-s}\|\, .
\end{split}
\end{equation*}

\appendix

\section{}\label{appendix}

In this section, we establish several lemmata that are useful for
the proof of the Mourre estimate in Section~\ref{S-Mourre}.

\begin{lemma}\label{lem:5.1}
Suppose that the kernels $G^{(\alpha)}$ satisfy
Hypotheses~\ref{hypothesis:2.1} and \ref{hypothesis:3.1}(ii). Then,
there exists a constant $D_1(G)$ such that, for $g\leq g_2$ and
$n\geq 1$,
\begin{equation}\nonumber
 | E - E_n| \leq g D_1(G) \sigma_n\, .
\end{equation}
\end{lemma}
%
%
%
\begin{proof}
For $g\leq g_2$, where $0<g_2$ is given by
\cite[Theorem~3.3]{BarbarouxGuillot2009}, we consider $\phi$
(respectively $\phi_n$), the unique normalized ground state of $H$
(respectively $H_n$) (see again
\cite[Theorem~3.3]{BarbarouxGuillot2009}). We thus have
\begin{equation}\label{eq:5.27}
\begin{split}
 & E - E_n \leq (\phi_n, (H-H_n)\phi_n)\, , \\
 & E_n - E \leq (\phi, (H_n - H) \phi )\, ,
\end{split}
\end{equation}
with
\begin{equation}\label{eq:5.28}
 H - H_n = g H_I (\chi_{\sigma_n/2}(p_2) G)\,
\end{equation}
where
\begin{equation}\nonumber
\chi_{\sigma_n/2}(p) =
\chi_0\left(\frac{|p|}{\sigma_n/2}\right)\, ,
\end{equation}
and $\chi_0(.) \in C^\infty(\R,\, [0,\,1])$ is fixed.
Combining \eqref{eq:2.53} of Lemma~\ref{lem:5.5} with
\eqref{eq:3.26} and \eqref{eq:3.28}, we get, together with
\eqref{eq:5.28}, for $g\leq g_2$,
\begin{equation}\nonumber
 \| (H-H_n)\phi_n\| \leq g K(\chi_{\sigma_n/2}(p_2) G)
 \left( \cb \|H_0\phi_n\| + \bb\right)\, ,
\end{equation}
and
\begin{equation}\nonumber
 \| (H - H_n) \phi \| \leq g K( \chi_{\sigma_n/2}(p_2) G)
   \left( \cb \|H_0 \phi\| + \bb \right)\, .
\end{equation}
It follows from Hypothesis~\ref{hypothesis:3.1}\textit{(ii)},
\cite[(4.9)]{BarbarouxGuillot2009}, and with \eqref{eq:5.18} that
there exists a constant $D_1(G)>0$ depending on $G$, such that,
for $n\geq 1$ and $g\leq g_2$, 
\begin{equation}
 \sup \left( \|(H-H_n)\phi_n\|,\
 \| (H-H_n)\phi\|\right) \leq g D_1(G) \sigma_n\, .
\end{equation}
By \eqref{eq:5.27}, this proves Lemma~\ref{lem:5.1}.
\end{proof}

\begin{lemma}\label{lem:5.2}
 We have
\begin{equation}\label{eq:5.32}
\begin{split}
 & \| \d\Gamma( (\cnt)^2 w^{(2)})
 \left(H_n - E_n - z\sigma_n\right )^{-1}\|
 \\
 & \leq
 \| (H_n - E_n) (H_n - E_n - z\sigma_n)^{-1}\| \leq
 1 + \frac{|z|}{ | \Im z |}\, .
\end{split}
\end{equation}
\end{lemma}
%
%
%
\begin{proof}
We have
\begin{equation}\label{eq:5.33}
 \1\otimes \d\Gamma((\cnt)^2 w^{(2)} ) \leq \1\otimes
 H_{0,n}^{(2)} \leq H_n - E_n\, .
\end{equation}
Set
 $$
 M_1 = \1 \otimes H_{0,n}^{(2)},\quad
 M_2 = (H^n - E^n)\otimes \1, \quad\mbox{and}\quad
 M = M_1 + M_2 = H_n - E_n\, ,
 $$
and let $\psi$ be in the algebraic tensor product
$\gD(M_1)\hat\otimes\gD(M_2)$. We obtain
\begin{equation*}
\begin{split}
 & \| (M_1\otimes \1 + \1\otimes M_2)\psi\|^2 \\
 & = \| (M_1\otimes \1)\psi\|^2 + \|(\1 \otimes M_2)\psi\|^2
 + 2\, \Re (\psi, \, (M_1\otimes\1)(\1 \otimes M_2) \psi) \\
 & = \| (M_1\otimes \1)\psi\|^2 + \|(\1 \otimes M_2)\psi\|^2
 + 2 ( (M_1{}^\frac12\otimes\1)\psi,\,
 (\1\otimes M_2)\, (M_1{}^\frac12\otimes\1)\psi) \\
 & \geq \| (M_1\otimes\1)\psi\|^2\, .
\end{split}
\end{equation*}
Thus, we get
\begin{equation}\label{eq:5.42}
 \| \d\Gamma\big( (\cnt)^2 w^{(2)}\big) \psi\|
 \leq \| (H_n - E_n) \psi\|\, .
\end{equation}
The set $\gD(M_1)\hat\otimes\gD(M_2)$ is a core for $M$, thus
\eqref{eq:5.42} is satisfied for every $\psi\in\cD(H_n - E_n) =
\gD(H_0)$. Setting
 $$
  \psi = (H_n - E_n - z\sigma_n)^{-1}\phi\, ,
 $$
in \eqref{eq:5.42}, we immediately get \eqref{eq:5.32}.
\end{proof}

\begin{lemma}\label{lem:5.3}
Suppose that the kernels $G^{(\alpha)}$ verify
Hypotheses~\ref{hypothesis:2.1} and ~\ref{hypothesis:3.1}(ii).
Then, there exists a constant $\tilde{C}_6^f(G)>0$ such that,
for $g\leq g_2$ and $n\geq 1$, 
\begin{equation}\label{eq:5.44}
 \| f_n(H_n - E_n) - f_n(H-E)\| \leq g\, \tilde{C}_6^f(G)\, .
\end{equation}

\end{lemma}
%
%
%
\begin{proof}
We have
\begin{equation}\label{eq:5.45}
\begin{split}
 & f_n(H_n-E_n)-f_n(H-E)\\
 &= \sigma_n \int \frac{1}{H_n - E_n - z\sigma_n}
 (H_n - H + E - E_n)\frac{1}{ H - E - z\sigma_n}
 \, \d\tilde{f}(z)\, .
\end{split}
\end{equation}
Combining \eqref{eq:2.53} of Lemma~\ref{lem:5.5}, \eqref{eq:3.26},
\eqref{eq:3.28}, and Hypothesis~\ref{hypothesis:3.1}\textit{(ii)},
we obtain, for every $\psi\in\cD(H_0)$ and for $g\leq g_2$,
\begin{equation}
 g \| H_I(\chi_{\sigma_n/2} G)\psi\|
 \leq
 g \sigma_n \tilde{K}(G) (\cb \| (H_0+1)\psi\|
 + (\cb+\bb)\|\psi\| )\, .
\end{equation}
This yields
\begin{equation}\label{eq:5.47}
 g \|H_I (\chi_{\sigma_n/2} G)  (H_0+1)^{-1}\|
 \leq g D_2(G) \sigma_n\, ,
\end{equation}
for some constant $D_2(G)$ and for $g\leq g_2$. Combining
Lemma~\ref{lem:5.1} with \eqref{eq:5.22} and
\eqref{eq:5.45}-\eqref{eq:5.47}, we obtain, for $g\leq g_2$,
\begin{equation}\label{eq:5.48}
 \| f_n(H_n -E_n) - f_n(H-E)\|
 \leq g D_2(G) \tilde{C}_3(G)
 \int \frac{\left| \frac{\partial \tilde{f}}{\partial \bar
 z}(x+iy)\right|} {y^2} (1 + |z| m_1) \d x \d y \, .
\end{equation}
Using \eqref{eq:5.11} and \eqref{eq:5.12}, we conclude the proof
of Lemma~\ref{lem:5.3} with
$$
 \tilde{C}_6^f(G) = D_2(G) \tilde{C}_3(G)
 \int \frac{\left| \frac{\partial \tilde{f}}{\partial \bar
 z}(x+iy)\right|} {y^2} (1 + |z| m_1) \d x \d y \, .
$$
\end{proof}

\begin{lemma}\label{lem:5.4}
Suppose that the kernels $G^{(\alpha)}$ satisfy
Hypotheses~\ref{hypothesis:2.1} and \ref{hypothesis:3.1}(ii). Then,
there exists a constant $\tilde{C}_7^f(G)>0$ such that, for $g\leq
g_2$ and $n\geq 1$,
\begin{equation}\nonumber
 \| \d\Gamma( (\cnt)^2 w^{(2)})
 ( f_n(H_n - E_n) - f_n (H - E))\|
 \leq g\, \tilde{C}_7^f \sigma_n\, .
\end{equation}
\end{lemma}
%
%
%
\begin{proof}
We have
\begin{equation}\nonumber
\begin{split}
 & \d\Gamma( (\cnt)^2 w^{(2)})(f_n(H_n-E_n) -
 f_n(H-E)) \\
 & = \sigma_n \int \d\Gamma( (\cnt)^2w^{(2)})
 \frac{1}{H_n - E_n - z\sigma_n}
 (H_n - H + E_n -E)
 \frac{1}{H-E-z\sigma_n} \d\tilde{f}(z)\, .
\end{split}
\end{equation}
Combining Lemmata~\ref{lem:5.1} and \ref{lem:5.2} with
\eqref{eq:5.22} and \eqref{eq:5.45}-\eqref{eq:5.47}, we obtain
\begin{equation}\nonumber
\begin{split}
 & \|\d\Gamma\left( (\cnt)^2 w^{(2)}\right)
 \left( f_n(H_n-E_n) - f_n(H-E)\right) \|\\
 & \leq
 g D_2(G) \tilde{C}_3(G) \sigma_n
 \int \left| \frac{\partial \tilde{f}}{\partial \bar z}
 (x+iy)\right|
 \left( 1 +\frac{|z|}{y}\right)
 \left( \frac{1 + |z| m_1}{y}\right) \d x \d y\, .
\end{split}
\end{equation}
Using \eqref{eq:5.11} and \eqref{eq:5.12}, we conclude the proof
of Lemma~\ref{lem:5.4} with
\begin{equation}\nonumber
 \tilde{C}_7^f(G) = D_2(G) \tilde{C}_3(G)
 \int \left|\frac{\partial \tilde{f}}{\partial \bar z}
 (x+iy)\right|
 \left( 1 +\frac{|z|}{y}\right)
 \left( \frac{1 + |z| m_1}{y}\right) \d x \d y\, .
\end{equation}
\end{proof}

The following lemma was proved in
\cite[(2.53)-(2.54)]{BarbarouxGuillot2009}, and gives explicitly
the relative bound for $H_I$ with respect to $H_0$. Note that this
bound holds for any interaction operator $H_I$ of the form
\eqref{eq:HI-quadraticform-1}-\eqref{eq:HI-quadraticform-3}, as
soon as the kernels $G^{(\alpha)}$ fulfil
Hypothesis~\ref{hypothesis:2.1}.
\begin{lemma}\label{lem:5.5}
For any $\eta>0$, $\beta>0$, and $\psi\in\cD(H_0)$, we have
\begin{equation}\label{eq:2.53}
\begin{split}
 & \| H_I \psi\| \\
 & \leq
 6 \sum_{\alpha=1,2} \| G^{(\alpha)}\|^2
 \left(
 \frac{1}{2m_W} \left(\frac{1}{m_1^2} +1\right) + \frac{\beta}{2m_W
 m_1^2} +\frac{2\eta}{m_1{}^2}(1+\beta)\right)\|H_0\psi\|^2 \\
 &
 + \left
 (\frac{1}{2m_W}\left(1+\frac{1}{4\beta}\right)
 + 2\eta\left(1+\frac{1}{4\beta}\right)
 + \frac{1}{2\eta}
 \right)
 \|\psi\|^2\, .
\end{split}
\end{equation}
\end{lemma}

\vspace{1cm}

\section{}\label{appendix}

\noindent\textbf{Intervals/sequences}

\noindent$\gamma=1 - \delta/(2m_1-\delta)$

\noindent$\sigma_0=2m_1+1$, $\sigma_1= m_1-\delta/2$,
$\sigma_{n+1}=\gamma\sigma_n$ ($n\geq 1$)

\noindent$\tau = 1 -\delta/(2 (2m_1 - \delta))$

\medskip
\noindent\textbf{Functions}

\noindent$w^{(1)}(\xi_1) = (|p_1|^2 + m_1{}^2)^\frac12,\quad
w^{(2)}(\xi_2) = |p_2|,\quad w^{(3)}(\xi_3) = (|k|^2 +
m_W{}^2)^\frac12$

\noindent$\chi_0\in C^\infty(\R,\, [0,1])$, $\chi_0=1$ on
$(-\infty,\, 1]$.

\noindent$\chi_\sigma(p)= \chi_0(|p|/\sigma), \quad
\tilde{\chi}^\sigma(p) = 1 - \chi_\sigma(p)$

\noindent$\chi^{(\tau)}(\lambda) =
 \left\{
  \begin{array}{ll}
     1 & \mbox{ for } \lambda\in (-\infty,\tau]\\
     0 & \mbox{ for } \lambda \in[1,\infty)
  \end{array}
 \right.
$

\noindent$\chi_n^{(\tau)}(p_2) =
\chi^{(\tau)}\left(|p_2|/\sigma_n\right)$

\noindent$
 f(\lambda)= \left\{
  \begin{array}{ll}
             1 & \mbox{if }\lambda\in[(\gamma -\epsilon_\gamma)^2,\, \gamma+\epsilon_\gamma] \\
             0 & \mbox{if }\lambda >\gamma + 2\epsilon_\gamma\\
             0 & \mbox{if } \lambda<(\gamma-2\epsilon_\gamma)^2
  \end{array}
 \right. ,\quad f_n(\lambda) = f(\lambda/\sigma_n)
$

\medskip
\noindent\textbf{Spaces}

\noindent$\begin{array}{ll}
 \Sigma_1 = \R^3\times\{-1/2,1/2\}
   &\Sigma_2 = \R^3\times\{-1,0,1\}\\
 \Sigma_1^{\sigma} = \Sigma_1\cap \{(p_2,\, s_2);\,
|p_2|\geq\sigma_2\},
   & \Sigma_{1,\sigma} = \Sigma_1\cap \{(p_2,\, s_2);\, |p_2|
   <\sigma\} \\
 \Sigma_n^{n+1} = \Sigma_1\cap\{(p_2,\,s_2);\ \sigma_{n+1}\leq
|p_2|\leq \sigma_n\} & \\
 \mbox{Electron Fock space: }\gF_a(L^2(\Sigma_1)) & \\
 \mbox{Neutrino Fock space: }\gF_a(L^2(\Sigma_1)) & \\
 \mbox{Boson Fock space: }\gF_s(L^2(\Sigma_2)) & \\
 \gF_2^\sigma = \gF_a(L^2(\Sigma_1^\sigma)),
   &\gF_{2,\sigma} = \gF_a(L^2(\Sigma_{1,\sigma})) \\
 \gF^\sigma = \gF_a(L^2(\Sigma_1))\otimes \gF_2^\sigma \otimes \gF_W,
   &\gF_\sigma = \gF_{2,\sigma}\\
 \gF_n^{n+1} = \gF_a(L^2(\Sigma_n^{n+1})) &
   \end{array}$

\medskip
\noindent\textbf{Hamiltonians}

\noindent$
 \begin{array}{ll}
  H_0^{(2)\sigma}= \int_{|p_2|>\sigma} w^{(2)}(\xi_2)\, c^*(\xi_2) c(\xi_2)
 \d\xi_2 &
  H_{0,\sigma}^{(2)} = \int_{|p_2|\leq\sigma} w^{(2)}(\xi_2)\, c^*(\xi_2) c(\xi_2)
 \d\xi_2 \\
  H_0^\sigma=H_0|_{\gF^\sigma}  \\
 H_{I,\sigma}(G) = H_I(\tilde{\chi}^\sigma(p_2) G) & \\
 H_\sigma = H_0 + gH_{I,\sigma} &
 H^\sigma = H_0^{(1)} + H_{0}^{(2) \sigma} + H_0^{(3)} + gH_{I,\sigma} \\
 H_n = H_{\sigma_n} & H^n = H^{\sigma_n} \\
 H_0^n = H_0^{\sigma_n} & \\
 H_+^n = H^n - E^n = H^n -\inf\mathrm{Spec} (H^n) & \\
 H_{0,n}^{n+1}= \displaystyle\int_{\sigma_{n+1}\leq |p_2| <\sigma_n}
 \!\!\!\!\!\!\!\!w^{(2)}(\xi_2) c^*(\xi_2) c(\xi_2) \d \xi_2 & \\
 \tilde{H}_+^n = H_+^n\otimes \1_n^{n+1} + \1_n\otimes H_{0,n}^{n+1} & \\
 H_{I,n}^{n+1} = H_I\big((\tilde{\chi}^{\sigma_{n+1}}(p_2) -
 \tilde{\chi}^{\sigma_{n}}(p_2)) G )&
\end{array}
$

\medskip
\noindent\textbf{Conjugate operators}

\noindent$a_n^{(\tau)} = \chi_n^{(\tau)}(p_2) \frac12\left(
p_2\cdot i\nabla_{p_2} + i\nabla_{p_2}\cdot
p_2\right)\chi_n^{(\tau)}(p_2)$

\noindent$A_n^{(\tau)} = \1 \otimes
\d\Gamma(a_n^{(\tau)})\otimes\1$


\section*{Acknowledgements}
The work was done partially while J.-M.~B. was visiting the
Institute for Mathematical Sciences, National University of
Singapore in 2008. The visit was supported by the Institute. The
authors also gratefully acknowledge the Erwin Schr\"odinger
International Institute for Mathematical Physics, where part of
this work was done. W.H.~A. is supported by the German Research
Foundation (DFG). J.-M.~B. gratefully acknowledges financial 
support from Agence Nationale de la Recherche, 
via the project HAM-MARK ANR-09-BLAN-0098-01



\begin{thebibliography}{10}



\bibitem{Ammari2004}
Z.~Ammari.
\newblock Scattering theory for a class of fermionic Pauli-Fierz
model.
\newblock {\em J. Funct. Anal.}, 208(2) (2004), 302--359.


\bibitem{Amouretal2007}
L.~Amour, B.~Gr{\'e}bert, and J.-C.~Guillot.
\newblock A mathematical model for the {F}ermi weak interactions.
\newblock {\em Cubo}, 9(2) (2007), 37--57.


\bibitem{Amreinetal1996}
W.~O.~Amrein, A.~Boutet~de~Monvel, and V.~Georgescu.
\newblock {\em {$C\sb 0$}-groups, commutator methods and spectral theory of
  {$N$}-body {H}amiltonians}, volume 135 of {\em Progress in Mathematics}.
\newblock Birkh\"auser Verlag, Basel, 1996.


\bibitem{Bachetal2006}
V.~Bach, J.~Fr{\"o}hlich, and A.~ Pizzo.
\newblock Infrared-finite algorithms in {QED}: the groundstate of an atom
  interacting with the quantized radiation field.
\newblock {\em Comm. Math. Phys.}, 264(1) (2006), 145--165.


\bibitem{Bachetal1999S}
V.~Bach, J.~Fr{\"o}hlich, and I.~M.Sigal.
\newblock Spectral analysis for systems of atoms and molecules coupled to the
  quantized radiation field.
\newblock {\em Comm. Math. Phys.}, 207(2) (1999), 249--290.


\bibitem{Barbarouxetal2004a}
J.-M.~Barbaroux, M.~Dimassi, and J.-C.~Guillot.
\newblock {Quantum electrodynamics of relativistic bound states with
cutoffs II.}
\newblock {\em Contemporary Mathematics}, 307(6) (2002), 9--14.


\bibitem{Barbarouxetal2004}
J.-M.~Barbaroux, M.~Dimassi, and J.-C.~Guillot.
\newblock Quantum electrodynamics of relativistic bound states with cutoffs.
\newblock {\em J. Hyperbolic Differ. Equ.}, 1(2) (2004), 271--314.


\bibitem{BarbarouxGuillot2009a}
J.-M.~Barbaroux and J.-C.~Guillot,
\newblock Limiting absorption principle at low energies for a
mathematical model of weak interactions: the decay of a boson.
\newblock {\em  C. R. Acad. Sci. Paris, Ser. I} 347, Issues 17-18 (2009), 1087--1092.


\bibitem{BarbarouxGuillot2009}
J.-M.~Barbaroux and J.-C.~Guillot,
\newblock Spectral theory for a
mathematical model of the weak interaction: The decay of the
intermediate vector bosons $\textbf{\textit{W}}^\pm$. I,
\newblock {\em Advances in Mathematical Physics} ID 978903 (2009),
and arXiv:0904.3171


\bibitem{Chenetal2009}
T.~Chen, J.~Faupin, J.~Fr\"ohlich and I.M.~Sigal,
\newblock Local decay in non-relativistic QED
\newblock arXiv:0911.0828, (2009).


\bibitem{DerezinskiGerard1999}
J.~Derezi{\'n}ski and C.~G{\'e}rard.
\newblock Asymptotic completeness in quantum field theory. {M}assive
{P}auli-{F}ierz {H}amiltonians.
\newblock {\em Rev. Math. Phys.}, 11(4) (1999), 383--450.


\bibitem{Derezinski2006}
J.~Derezi{\'n}ski,
\newblock {\em Introduction to representations of the canonical
commutation and anticommutation relations},
\newblock Large {C}oulomb systems,
\newblock Lecture Notes in Phys. 695, 63--143
\newblock Springer, Berlin, 2006.


\bibitem{DimassiGuillot}
M.~Dimassi and J.C.~Guillot,
\newblock {\em The quantum electrodynamics of relativistic states with cutoffs. I},
\newblock {\em Appl. Math. Lett.} \textbf{16}, no.4 (2003) , 551--555.


\bibitem{Frohlichetal2008}
J.~Fr{\"o}hlich, M.~Griesemer, and I.~M. Sigal.
\newblock Spectral theory for the standard model of non-relativistic {QED}.
\newblock {\em Comm. Math. Phys.} 283(3) (2008), 613--646.


\bibitem{GeorgescuGerard} V. Georgescu, C. G\'erard,
\newblock On the virial theorem in quantum mechanics
\newblock {\em Comm. Math. Phys.} 208 (1999), 275--281.


\bibitem{Gerard2008} C. G\'erard,
\newblock A proof of the abstract limiting absorption principles
by energy estimates
\newblock {\em J. Funct. Anal.} 254 (2008), no.~11, 2707--2724.


\bibitem{Georgescuetal2004}
V.~Georgescu, C.~G\'erard, J.~S.~M\o ller.
\newblock Spectral theory of massless Pauli-Fierz models.
\newblock {\em Comm. Math. Phys.}, 249(1) (2004), 29--78.


\bibitem{GoleniaJecko2007} S. Gol\'enia, T. Jecko,
\newblock A new look at Mourre's commutator theory,
\newblock {\em Complex and Oper. Theory} 1 (2007), no~3, 399--422.


\bibitem{GreinerMuller1989}
W.~Greiner and B.~M{\"u}ller.
\newblock {\em Gauge Theory of weak interactions}
\newblock Springer, Berlin, 1989.


\bibitem{GLL}
M.~Griesemer, E.H.~Lieb and M.~Loss,
\newblock Ground states in non-relativistic quantum electrodynamics,
\newblock {\em Inv. Math} 145 (2001), 557--595.


\bibitem{Kato1966}
T.~Kato.
\newblock {\em Perturbation Theory for Linear Operators},
volume 132 of {\em Grundlehren der mathematischen
{W}issenschaften}.
\newblock Springer-Verlag, Berlin, 1 edition, 1966.


\bibitem{ReedSimon1975}
M.~Reed and B.~Simon.
\newblock {\em Methods of modern mathematical physics. {II}. {F}ourier
  analysis, self-adjointness}.
\newblock Academic Press, New York, 1975.


\bibitem{Sahbani1997}
J.~Sahbani.
\newblock The conjugate operator method for locally regular {H}amiltonians.
\newblock {\em J. Operator Theory}, 38(2) (1997), 297--322.


\bibitem{WeinbergI2005}
S.~Weinberg.
\newblock {\em The quantum theory of fields. {V}ol. {I}}.
\newblock Cambridge University Press, Cambridge, 2005.
\newblock Foundations.


\bibitem{WeinbergII2005}
S.~Weinberg.
\newblock {\em The quantum theory of fields. {V}ol. {II}}.
\newblock Cambridge University Press, Cambridge, 2005.
\newblock Modern applications.

\end{thebibliography}
\end{document}